%% file: paper.tex
\documentclass[acmsmall,nonacm,screen]{acmart}
\input{prelude}
\input{examples}
\setcopyright{acmcopyright}
\copyrightyear{2018}
\acmYear{2018}
\acmDOI{XXXXXXX.XXXXXXX}

\acmPrice{15.00}
\acmISBN{978-1-4503-XXXX-X/18/06}

\begin{document}

\title{Asynchronous Modal FRP}


\author{Patrick Bahr}

\affiliation{
  \institution{IT University of Copenhagen}            
  \country{Denmark}                    
}
\email{paba@itu.dk}          

\author{Rasmus Ejlers M{\o}gelberg}

\affiliation{
  \institution{IT University of Copenhagen}            
  \country{Denmark}                    
}
\email{mogel@itu.dk}          


\begin{abstract}
  Over the past decade, a number of languages for functional reactive
  programming (FRP) have been suggested, which use modal types to
  ensure properties like causality, productivity and lack of space
  leaks. So far, almost all of these languages have included a modal
  operator for delay on a global clock. For some applications,
  however, a global clock is unnatural and leads to leaky
  abstractions as well as inefficient implementations. While modal
  languages without a global clock have been proposed, no operational
  properties have been proved about them, yet.

  This paper proposes Async RaTT, a new modal language for
  asynchronous FRP, equipped with an operational semantics mapping
  complete programs to machines that take asynchronous input signals
  and produce output signals. The main novelty of Async RaTT is a new
  modality for asynchronous delay, allowing each output channel to be
  associated at runtime with the set of input channels it depends on,
  thus causing the machine to only compute new output when
  necessary. We prove a series of operational properties including
  causality, productivity and lack of space leaks. We also show that,
  although the set of input channels associated with an output channel
  can change during execution, upper bounds on these can be determined
  statically by the type system.
\end{abstract}



\maketitle

\section{Introduction}

Reactive programs are programs that engage in a dialogue with their environment,
receiving input and producing output, often without ever
terminating. Examples include much of the most safety critical
software in use today, such as control software and servers, as well
as GUIs. Most reactive software is written in imperative languages
using a combination of complex features such as callbacks and shared
memory, and for this reason it is error-prone and hard to reason
about.

The idea of functional reactive programming (FRP)~\cite{FRAN} 
is to provide the programmer with the
right abstractions to write reactive programs in functional style, allowing for short 
modular programs, as well as modular reasoning about these.  For such abstractions
to be useful it is important that they are designed to allow for efficient low-level 
implementations to be automatically generated from programs. 

The main abstraction of FRP is that of signals, which are
time-dependent values.  In the case of discrete time given by a global
clock, a signal can be thought of as a stream of data. A reactive
program is essentially just a function taking input signals and
producing output signals. For this to be implementable, however, it
needs to be causal: The current output must only depend on current
and past input.  Moreover, the low-level implementations generated
from high-level programs should also be free of (implicit) space- and
time-leaks. This means that reactive programs should not store data
indefinitely causing the program to eventually run out of space, nor
should they repeat computations in such a way that the execution of
each step becomes increasingly slower.

These requirements have led to the development of modal
FRP~\cite{jeffrey2014,jeffrey2012,krishnaswami2011ultrametric,krishnaswami2012higher,krishnaswami13frp,bahr2019simply,rattusJFP,Jeltsch2012},
a family of languages using modal types to ensure that all programs
can be implemented efficiently.  The most important modal type
constructor is $\Delay$, used to classify data available in the next
time step on some global discrete clock. For example, the type of signals
should satisfy the type isomorphism
$\Sig\,A \iso A \times \Delay (\Sig\, A)$ stating that the current
value of the signal is available now, but its future values are only
available after the next time step. Using this encoding of signals,
one can ensure that all reactive programs are causal. Many modal FRP
languages also include a variant of the \citet{nakano2000} guarded
fixed point operator of type $(\Delay A \to A) \to A$. The type ensures
that recursive calls are only performed in future steps, thus ensuring 
termination of each step of computation, a property called \emph{productivity}.
Often these
languages also include a $\Box$ modality used to classify data that is
\emph{stable}, in the sense that it can be kept until the next time
step without causing space leaks. Other modal constructors, such as
$\Diamond$ (eventually) can be encoded, suggesting a Curry-Howard
correspondence between linear temporal logic~\cite{LTL} and modal
FRP~\cite{jeffrey2012,Jeltsch2012,cave14fair,bahr2021diamonds}.

However, for many applications, the notion of a global clock associated
with the $\Delay$ modal operator may not be natural and can also lead to 
inefficient implementations.
%
%
Consider, for example, a GUI which takes an input signal of user
keystrokes, as well as other signals that are updated more frequently,
like the mouse pointer coordinates. The global clock would have to tick at
least as fast as the updates to the fastest signal, and updates on the
keystroke signal will only happen on very few ticks on the global
clock. Perhaps the most natural way to model the keystroke signal is
therefore using a signal of type $\sym{Maybe(Char)}$. In the modal FRP
languages of \citet{krishnaswami13frp,bahr2019simply}, the processor
for this signal will have to wake up for each tick on the global
clock, check for input, and often also transport some local state to
the next time step by calling itself recursively. Perhaps more
problematic, however, is that an important abstraction barrier is
broken when a processor for an input signal is given access to the global clock.
Instead, we would like to write the GUI as a collection of processors
for asynchronous input signals that are only activated upon updates to
the signals on which they depend.

\subsection{Async RaTT}
\label{sec:async-ratt}
This paper presents Async RaTT, a modal FRP language in the RaTT
family \cite{bahr2019simply,rattusJFP,bahr2021diamonds}, designed for
processing asynchronous input. A reactive program in Async RaTT reads
signals from a set of \emph{input channels} and in response sends
signals to a set of \emph{output channels}. In a GUI application,
typical input channels would include the mouse position and keystroke
events, while output channels could for example include the content of
a text field or the colour of a text field.

\begin{figure}[t]
  \begin{subfigure}{0.45\textwidth}
    \begin{center}
  \begin{tikzpicture}
    \node (k1) at (-2,0.8) {$\kappa_1$};
    \node (k2) at (0,0.8) {$\kappa_2$};
    \node (k3) at (2,0.8) {$\kappa_3$};
    \begin{scope}[every node/.style={draw,minimum size=4mm }]
    \node (n1) at (-2,0) {};
    \node (n2) at (0,0) {};
    \node (n3) at (2,0) {};
    \node (m1) at (-1,-1) {};
    \node (m2) at (1,-1) {};
    \end{scope}
    \node[align=center] (o1) at (-1,-2) {$o_1$\\\small clock $\set{\kappa_1}$};
    \node[align=center] (o2) at (1,-2) {$o_2$\\\small clock $\set{\kappa_2,\kappa_3}$};

    \draw[-latex]%
    (k1) edge (n1)%
    (k2) edge (n2)%
    (k3) edge (n3)%
    (n1) edge (m1)%
    (n2) edge (m2)%
    (n3) edge (m2)%
    (m1) edge (o1)%
    (m2) edge (o2)%
    ;%
  \end{tikzpicture}
    \end{center}
    \caption{Dataflow graph with its computed clocks.}
  \label{fig:dataflowA}  
\end{subfigure}
  \begin{subfigure}{0.45\textwidth}
    \begin{center}
  \begin{tikzpicture}
    \node (k1) at (-2,0.8) {$\kappa_1$};
    \node (k2) at (0,0.8) {$\kappa_2$};
    \node (k3) at (2,0.8) {$\kappa_3$};
    \begin{scope}[every node/.style={draw,minimum size=4mm }]
    \node (n1) at (-2,0) {};
    \node (n2) at (0,0) {};
    \node (n3) at (2,0) {};
    \node (m1) at (-1,-1) {};
    \node (m2) at (1,-1) {};
    \end{scope}
    \node[align=center] (o1) at (-1,-2) {$o_1$\\\small clock $\set{\kappa_1,\kappa_2}$};
    \node[align=center] (o2) at (1,-2) {$o_2$\\\small clock $\set{\kappa_2,\kappa_3}$};

    \draw[-latex]%
    (k1) edge (n1)%
    (k2) edge (n2)%
    (k3) edge (n3)%
    (n1) edge (m1)%
    (n2) edge (m1)%
    (n2) edge (m2)%
    (n3) edge (m2)%
    (m1) edge (o1)%
    (m2) edge (o2)%
    ;%
  \end{tikzpicture}
    \end{center}
    \caption{Clocks are updated as dependencies change.}
  \label{fig:dataflowB}  
  \end{subfigure}
  \caption{Dynamically changing dataflow graph of an Async RaTT
    program with input channels $\kappa_1,\kappa_2,\kappa_3$ and
    output channels $o_1,o_2$.}
  \label{fig:dataflow}
\end{figure}

For each output channel $o$, the reactive program keeps track of the
set $\theta$ of input channels on which $o$ depends (cf.\
\autoref{fig:dataflowA}). We refer to such a set $\theta$ of input
channels as a \emph{clock}.  When the signal on an input channel
$\kappa$ is updated, only those output channels whose clock $\theta$
contains $\kappa$ will be updated. For example, the keystroke input
channel might be in the clock for the text field content but not the
text field colour. Since the program can dynamically change its
internal dataflow graph, the clock associated with an output channel
may change during execution (cf.\ \autoref{fig:dataflowB}) and so is
not known at compile time. For example, the text field might fall out
of focus and thus not react to keystrokes any longer. We refer to the
arrival of new data on an input channel in the clock $\theta$ as a
\emph{tick} on clock $\theta$.

Async RaTT has a modal operator $\Box$ used to classify stable data,
as well as two new modalities: $\DelayE$ for asynchronous delays and
$\DelayA$ for a delay on the global clock.  A value of type
$\DelayE A$ is a pair consisting of a clock $\theta$ and a computation
that can be executed to return data of type $A$ on the first tick on
$\theta$. The type $\DelayE A$ can therefore be thought of as an
existential type. The clocks of output channels, as illustrated in
\autoref{fig:dataflow}, are stored in the first component of this
existential type.  Our notion of signal is encoded in types as a
recursive type $\Sig\,A \iso A \times \DelayE \Sig\,A$. That means,
the clock associated with the tail of a signal may change from one
step to the next, allowing for dynamic updates of clocks associated with
output channels as in \autoref{fig:dataflow}.
%


Unlike the synchronous $\Delay$, the asynchronous $\DelayE$ does
not have an applicative action of type $\DelayE(A \to B) \to \DelayE A \to \DelayE B$
because the delayed function and the delayed input may not 
arrive at the same time, and to avoid space leaks, Async RaTT
does not allow the first input to be stored until the second input arrives. 
Instead, Async RaTT synchronises delayed data using an operator
\begin{equation*}
  \synch : \DelayE A_1 \to \DelayE A_2 \to 
   \DelayE((A_1\times\DelayE A_2) + (\DelayE A_1 \times A_2) + (A_1 \times A_2))
\end{equation*}
Given two delayed computations associated with clocks $\theta_1$ and
$\theta_2$, respectively, $\synch$ returns the delayed computation
associated with the union clock $\theta_1 \sqcup \theta_2$. This
delayed computation waits for an input on any input channel
$\kappa \in \theta_1 \sqcup \theta_2$, and then evaluates the
computations that can be evaluated depending on whether
$\kappa \in \theta_1$, $\kappa \in \theta_2$, or both. For example, if
the input arrives on channel $\kappa \in \theta_1 \setminus\theta_2$,
only the first delayed computation is evaluated.  The $\synch$ operator 
can be used to implement operators like
\[
  \sym{switch} : \Sig\,A \to \DelayE(\Sig\, A) \to \Sig\,A
\]
which dynamically update the dataflow graph of a program. 

Note that $\synch$ can be read as a linear time axiom:
Given two clocks, either one ticks before the other, or they
tick simultaneously. Async RaTT programs are therefore
dependent on the run-time environment to schedule the order
in which inputs are processed. What we mean by asynchronicity 
is that output channels are updated asynchronously. This is
reflected in the type system by $\DelayE$ not being an applicative
functor as explained above. 

The modal type $\DelayA A$ classifies computations that
can be run at any time in the future, but not now.
It is used in the
guarded fixed point operator, which in Async RaTT has type
\[
 \Box(\DelayA A \to A) \to A
\]
The input to the fixed point operator must be a stable function
(as classified by $\Box$), because it will be used in unfoldings
at any time in the future. 
The use of $\DelayA$ restricts fixed points to only unfold in 
the future, ensuring termination of each step of computation.

\subsection{Operational Semantics and Results}

We present an operational semantics mapping each complete Async RaTT
program to a machine that transforms a sequence of inputs received on
its input channels to a sequence of outputs on its output
channels. The transformation is done in steps, processing one input at
a time, producing new outputs on the affected output channels.

The operational semantics consists of two parts. The first is the evaluation 
semantics describing the evaluation of a term in each step of the evaluation. This 
takes a term and a store and returns a value and an updated store in the context
of current values on input signals. The store contains delayed computations, and the
evaluation semantics may run previously stored delayed computations as well 
as store new ones to be evaluated at a later step. 
The reactive semantics on the other hand, describes the machine
which, at each step, locates the output signals to be updated and executes the 
corresponding delayed computations to produce output. 

The transformation of input to
output described by the operational semantics is causal by construction. We show that it is also
deterministic and productive (in the sense that each step terminates
and never gets stuck). We also show that the execution of an Async
RaTT program is free of (implicit) space leaks. This is achieved following a
technique originally due to \citet{krishnaswami13frp}: At the end of
each step of execution, the machine deletes all delayed computations
that in principle could have been run in the current step --
regardless of whether they actually were run. All inputs are also
deleted, either at the end of the step or when the next input from the
same signal arrives, depending on the kind of the specific input
signal. Our results show that this aggressive garbage collection
strategy is safe. Of course, the programmer can still write programs
that accumulate space, but such leaks will be explicit in the source
program, not implicitly introduced by the implementation of the
language. (See \citet{krishnaswami13frp} for a further discussion of
implicit vs explicit space leaks.)

Finally, we show that an upper bound on the dynamic clocks associated
with an output signal can be computed statically. More precisely,
given an Async RaTT program consisting of a number of output signals
in a given context $\Delta$ of input channels, if one of the output
signals can be typed in a smaller context $\Delta'\subseteq\Delta$,
then that signal will never need to update on input arriving on
channels in $\Delta\setminus\Delta'$. Note that this result holds
despite the existence of operators like $\sym{switch}$, which
dynamically change the dataflow graph of a program.


\subsection{Overview}

The paper is organised as follows: Async RaTT is presented along with
its typing rules in \autoref{sec:asynch:RaTT}, and
\autoref{sec:examples} illustrates the expressivity of Async RaTT by
developing a small library of signal combinators, along with examples
that use the library for GUI programming and computing integrals and
derivatives of signals.  The operational semantics is defined in
\autoref{sec:oper-semant}, which also illustrates it with an example,
and presents the main results. Section~\ref{sec:metatheory} sketches
the proofs of the main results, and in particular defines the Kripke
logical relation used for the proofs. Finally,
\autoref{sec:related-work} and \autoref{sec:concl-future-work} discuss
related work, conclusions and future work. In addition,
Appendix~\ref{sec:proof-fund-prop} gives a detailed account of the
proof of the fundamental property of the Kripke logical relation.

\section{Async RaTT}
\label{sec:asynch:RaTT}

\begin{figure}[t]
  \[
    \arraycolsep=2pt
    \begin{array}{llcl}
   
      \text{Locations} &l& \in & \locs\\
      \text{Input Channels} &\kappa& \in & \chans\\
      \text{Clock Expr.} &\theta& ::=\; & \cl{v}\mid \theta
                                                \sqcup \theta'\\
      \text{Types} &A,B &::=\; & \alpha \mid \Unit \mid \Nat \mid  A \times B
                                 \mid A + B \mid A \to B \mid
                                 \DelayE A\mid \DelayA A \mid \Fix\;\alpha.A
      \mid \Box A\\
      \text{Stable Types} &S, S' &::=\; & \Unit \mid \Nat \mid  S \times S'
                                 \mid S + S' \mid \DelayA A
			      \mid \Box A\\
      \text{Value Types} &T, T' &::=\; & \Unit \mid \Nat \mid  T \times T'
                                 \mid T + T' \\
      \text{Values} &v,w &::=\; &x \mid \unit \mid \zero \mid \suc \, v \mid \lambda x.t \mid
                                  \pair{v}{w} \mid \interm_i\, v \mid
                                  l\mid \inp \mid \rbox\,t\mid \dfix\,x .t\mid \into\,v\\
      \text{Terms} &s,t &::=\; & v \mid \suc\, t \mid \recN{s}{x}{y}{t}{u} \mid 
                                 \pair{s}{t} \mid \interm_i\, t
                                 \mid \pi_i\,t \mid t_1t_2  \mid
                                 \letterm x s t\\ &&\mid&                               
   \caseterm{t}{x}{t_1}{x}{t_2}\mid \delay_\theta\, t \mid \adv\,t \mid \select\, v_1\,v_2 \mid \unbox\,t \\ && \mid&
   \fix\,x.t \mid \never \mid \into\, t \mid \out\, t \mid \buf
    \end{array}
  \]
\caption{Syntax.}
\label{fig:syntax}
\end{figure}

\begin{figure}[th]
  \begin{mathpar}
    \inferrule*%
  {~}%
  {\wfcxt{\cdot}}%
  \and
  \inferrule*%
  {\wfcxt{\Gamma} \\ x \nin \dom\Gamma}%
  {\wfcxt{\Gamma,x : A}}%
  \and
  \inferrule*%
  {\wfcxt{\Gamma} \\ \isclock{\Gamma}{\theta}\\ \Gamma \text{ tick-free}}%
  {\wfcxt{\Gamma,\tick[\theta]}}%
  \and
  \inferrule*%
  {\isclock{\Gamma}{\theta} \\ \isclock{\Gamma}{\theta'}}%
  {\isclock{\Gamma}{\theta \sqcup \theta'}}%
  \and
  \inferrule*%
  {\hastype{\Gamma}{v}{\DelayE A}}%
  {\isclock{\Gamma}{\cl{v}}}%
    \and
  \inferrule*%
  {\Gamma' \text{ tick-free or } A \; \stable \\ \wfcxt{\Gamma,x:A,\Gamma'}}%
  {\hastype{\Gamma,x:A,\Gamma'}{x}{A}}%
  \and%
  \inferrule* {~}%
  {\hastype{\Gamma}{\unit}{\Unit}}%
  \and%
  \inferrule*%
  {\hastype{\Gamma}{s}{A} \\ \hastype{\Gamma,x:A}{t}{B}}%
  {\hastype{\Gamma}{\letterm x s t}{B}}%
  \and%
  \inferrule*%
  {\hastype{\Gamma,x:A}{t}{B} \\ \Gamma\text{ tick-free}}%
  {\hastype{\Gamma}{\lambda x.t}{A \to B}}%
  \and%
  \inferrule*%
  {\hastype{\Gamma}{t}{A \to B} \\
    \hastype{\Gamma}{t'}{A}}%
  {\hastype{\Gamma}{t\,t'}{B}}%
  \and%
  \inferrule*%
  {\hastype{\Gamma}{t}{A_i} \\ i \in \{1, 2\}}%
  {\hastype{\Gamma}{\interm_i\, t}{A_1 + A_2}}%
  \and%
  \inferrule*%
  {\hastype{\Gamma,x: A_i}{t_i}{B} \\
    \hastype{\Gamma}{t}{A_1 + A_2} \\ i \in \{1,2\}}%
  {\hastype{\Gamma}{\caseterm {t}{x}{t_1}{x}{t_2}}{B}}%
  \and
  \inferrule*%
  {\hastype{\Gamma}{t}{A} \\
    \hastype{\Gamma}{t'}{B}}%
  {\hastype{\Gamma}{\pair{t}{t'}}{A \times B}}%
  \and%
  \inferrule*%
  {\hastype{\Gamma}{t}{A_1 \times A_2} \\ i \in \{1, 2\}}%
  {\hastype{\Gamma}{\pi_i\,t}{A_i}}%
  \and%
  \inferrule*%
  {\hastype{\Gamma,\tick[\theta]}{t}{A} \\ \isclock{\Gamma}{\theta}}%
  {\hastype{\Gamma}{\delay_\theta\,t}{\DelayE A}}%
    \and
  \inferrule*%
  {~}%
  {\hastype{\Gamma}{\never}{\DelayE A}}%
  \and%
  \inferrule*%
  {\kappa :_\capa A \in \Delta \\ \capa \in \{\cpush, \cbp\}}%
  {\hastype{\Gamma}{\inp}{\DelayE A}}%
  \and%
  \inferrule*%
  {\kappa :_\capa A \in \Delta \\ \capa \in \{\cbuf, \cbp\}}%
  {\hastype{\Gamma}{\buf}{A}}%
  \and%
  \inferrule*%
  {\hastype{\Gamma}{v}{\DelayE A}\\ \wfcxt{\Gamma,\tick[\cl{v}],\Gamma'}}%
  {\hastype{\Gamma,\tick[\cl{v}],\Gamma'}{\adv\,v}{A}}%
  \and%
  \inferrule*%
  {\hastype{\Gamma}{v_1}{\DelayE A_1}
    \\\hastype{\Gamma}{v_2}{\DelayE A_2}\\
    \vdash \theta_1\sqcup\theta_2 = \cl{v_1} \sqcup \cl{v_2}\\\wfcxt{\Gamma,\tick[\theta_1\sqcup\theta_2],\Gamma'}}%
  {\hastype{\Gamma,\tick[\theta_1\sqcup\theta_2],\Gamma'}{\select\,v_1\,v_2}{((A_1
      \times \DelayE A_2) + (\DelayE A_1
      \times A_2)) + (A_1 \times A_2)}}%
  \and%
  \inferrule* {~}%
  {\hastype{\Gamma}{\zero}{\Nat}}%
  \and%
  \inferrule*%
  {\hastype{\Gamma}{t}{\Nat}}%
  {\hastype{\Gamma}{\suc \, t}{\Nat}}%
  \and%
  \inferrule*%
  {\hastype{\Gamma}{s}{A} \\
    \hastype{\Gamma,x:\Nat,y:A}{t}{A} \\ \hastype{\Gamma}{n}{\Nat}}
  {\hastype{\Gamma}{\recN{s}{x}{y}{t}{n}}{A}}%
  \and
    \inferrule*
  {\hastype{\stabilize\Gamma,x : \DelayA A}{t}{A}}
  {\hastype{\Gamma}{\fix\,x. t}{A}}
  \and
  \inferrule*%
  {\hastype{\Gamma}{x}{\DelayA A}\\ \wfcxt{\Gamma,\tick[\theta],\Gamma'}}%
  {\hastype{\Gamma,\tick[\theta],\Gamma'}{\adv\,x}{A}}%
  \and
  \inferrule*
  {\hastype{\stabilize\Gamma}{t}{ A}}
  {\hastype{\Gamma}{\rbox\, t}{\Box A}}
  \and
  \inferrule*
  {\hastype{\Gamma}{t}{\Box A}}
  {\hastype{\Gamma}{\unbox\, t}{A}}
    \and
  \inferrule*
  {\hastype{\Gamma}{t}{\Fix \;\alpha.A}}
  {\hastype{\Gamma}{\out\,t}{A[\DelayE(\Fix \; \alpha.A)/\alpha]}}
  \and
    \inferrule*
  {\hastype{\Gamma}{t}{A[\DelayE(\Fix\,\alpha.A)/\alpha]}}
  {\hastype{\Gamma}{\into\,t}{\Fix\,\alpha.A}}
\end{mathpar}
\begin{gather*}
  \stabilize\cdot = \cdot \qquad \stabilize{(\Gamma,\tick[\theta])} = \stabilize\Gamma
  \qquad \stabilize{(\Gamma,x:A)} =
  \begin{cases}
    \stabilize\Gamma, x : A &\text{if $A$ stable}\\
    \stabilize\Gamma &\text{otherwise}
  \end{cases}
\end{gather*}
\caption{Typing rules.}
\label{fig:typing}
\end{figure}

This section gives an overview of Async RaTT, referring to
Figures~\ref{fig:syntax} and~\ref{fig:typing} for the full
specification of its syntax and typing rules.

An Async RaTT program has access to a set of input channels, each of
which receive updates asynchronously from each other. To account for
this, typing judgements are relative to an input channel context
$\Delta$ or \emph{input context} for short. An example of such a
context is
\[
  \sym{keyPressed} :_{\cpush} \Nat ,\sym{mouseCoord} :_{\cbp} \Nat
  \times \Nat, \sym{time} :_{\cbuf} \Float
\]
There are three classes of input channels, each corresponding to one
of the subscripts $\cpush, \cbuf$, and $\cbp$ as in the example
above. \emph{Push-only} input channels, indicated by $\cpush$, are
input channels whose updates are pushed through the program, possibly
causing output channels to be updated.  In the example context above,
the programmer will want to react to user keypresses immediately, and
so updates to this should be pushed. On the other hand, we may wish to
have access to a time input channel, which we can read from at any
time, but we may not want the program to wake up whenever the time
changes.  Time is therefore treated as a \emph{buffered-only} input
channel, indicated by $\cbuf$, whose most recent value is buffered,
but whose changes will not trigger the program to update any output
channel. Finally, input channels may be both buffered and pushed,
indicated by $\cbp$, which means that updates are pushed, but we also
keep the value around in a buffer, so that the latest value can always
be read by the program. This is unlike the push-only input channels
whose values are deleted for space efficiency reasons, once an update
push has been treated. For example, we might want to be informed when
the mouse coordinates are updated, but also keep these around so that
we can read the mouse coordinates when a key is pressed, even if the
mouse has not moved. We refer to input channels that are either
push-only or buffered-push ($\cpush$ or $\cbp$) as \emph{push channels} and
similarly to input channels that are either buffered-only or
buffered-push as \emph{buffered channels}.

All signals are assumed to have value types, i.e., any declaration
$\kappa :_\capa A$ in $\Delta$ must have a value type $A$. The grammar
for value types is given in Figure~\ref{fig:syntax}.

\subsection{Clocks and $\DelayE$}

A clock is intuitively a set of push channels ($\cpush$ or $\cbp$),
that the program may have to react to. For instance, $\emptyset$,
$\{\sym{keyPressed}\}$ and $\{\sym{keyPressed} ,\sym{mouseCoord}\}$
are all examples of clocks for the example input context mentioned
earlier. The type $\DelayE A$ is a type of delayed computation on an
existentially quantified clock. In other words, a value of type
$\DelayE A$ is a pair of a clock $\theta$ and a computation that will produce a
value of type $A$ once an update on one of the input channels
in $\theta$ is received. We refer to such an update as a \emph{tick}
on the clock $\theta$. For example, if the associated
clock is $\{\sym{keyPressed} ,\sym{mouseCoord}\}$, then the data of
type $A$ can be computed once $\sym{keyPressed}$ or $\sym{mouseCoord}$
receive new input.

Since $\DelayE A$ are existential types, one can obtain the clock
$\cl v$ for any \emph{value} of these types. The values of type
$\DelayE A$ are variables and $\inp$ where $\kappa$ is one of the
push channels. The latter acts as a reference to the next value pushed on
$\kappa$, and so intuitively $\cl{\inp} = \{\kappa\}$. Clocks can also
be combined using a union operator $\sqcup$. We also include an
element $\never$ which is associated with the empty clock.

We use Fitch-style~\cite{clouston2018fitch}, 
rather than the more traditional dual context
style~\cite{davies2001modal} for programming with the modal type constructors of Async
RaTT. 
In the case
of $\DelayE$, this means that introduction and elimination rules use
a special symbol $\tick[\theta]$,  referred to as a \emph{tick}, in the context. 
One can think of a tick $\tick[\theta]$ as representing ticks of the clock $\theta$,
and it divides the judgement 
into variables (to the left of $\tick[\theta]$) received before the tick, and everything else, 
which happens after the tick. 
For example, the elimination rule should be read as: If $v$ has type
$\DelayE A$ now, then after a tick on the clock $\cl{v}$,
$\adv(v)$ has type $A$. 
Similarly, the introduction rule for $\DelayE$ should be read as: If 
$t$ has type $A$ after a tick on clock $\theta$ then $\delay_\theta\,t$
has type $\DelayE A$ now. 
Note that there can be at most one tick in a context. This is a restriction that 
is required for the proof of the productivity theorem (\autoref{thr:productivity}), 
and also appears in other languages in the RaTT 
family~\cite{bahr2019simply,bahr2021diamonds}. However,
\citet{rattusJFP} shows that this restriction can be lifted by a
program transformation that transforms a program typable with multiple
ticks into one with only one tick and where $\adv$ is only applied to
variables.

Operationally, the term $\delay_\theta\,t$
creates a delayed computation which is stored in a heap until the 
input data necessary for evaluating it is available. It is therefore not considered 
a value, rather 
$\delay_\theta\,t$ evaluates to a heap reference $l$ to the delayed computation.
Although heap references are part of Async RaTT, and even considered values 
(Figure~\ref{fig:syntax}),
programmers are not allowed to use these directly, and there are therefore no typing
rules for them. 
%

Two delayed values $v_1 : \DelayE A_1$ and $v_2 : \DelayE A_2$ can be
synchronised using $\select$ once a tick on the union       
clock $\cl{v_1} \sqcup \cl{v_2}$ has been received. 
The type of  $\select\,v_1\,v_2$ reflects the three possible cases
for such a tick: 
It could be in one $\cl{v_i}$, but not the other, or it could be in both. 
For example, if the input
is in $\cl{v_1}$, but not $\cl{v_2}$, then data of type
$A_1 \times \DelayE A_2$ can be computed. 
The $\synch$ operator shown in section~\ref{sec:async-ratt} can be defined using $\select$.
The idea of using a term
like $\select$ to distinguish between these cases is due to
\citet{graulund2021adjoint}, who only require two cases to be defined, 
resorting to non-deterministic choice in the case where the tick is in 
the intersection of the clocks. In Async RaTT, providing all three cases
is crucial for the operational results of \autoref{sec:oper-semant}.

%
Note that the rules for $\select$ and $\adv$ restrict the application of
these constructions to values. One reason for this is that it 
simplifies the metatheory by preventing
arbitrary terms occurring in contexts through clocks. It also means that
clock expressions always are values that do not need to be 
evaluated. For example, evaluating $\delay_{\cl t}(\adv(t))$
requires evaluating $t$ twice: 
first for evaluating the clock, and then to evaluate the term itself. 
Elimination of $\DelayE$ can be done for more general terms $t$ 
using a combination of let-binding and $\adv$.

\subsection{Stable Types and Fixed Points}

General values in Async RaTT can contain references to time-dependent
data, such as delayed computations stored in the heap. One of the main
purposes of the type system is to prevent such references to be
dereferenced at times in the future when a delayed computation has
been deleted from the heap. For this reason, arbitrary data should not
be kept across time steps, and this is reflected in the type system in
the variable introduction rule which prevents general variables to be
introduced across ticks.

For some types, however, values can not contain such references. We
refer to these as \emph{stable types} and the grammar for these is
given in Figure~\ref{fig:syntax}. Stable types include all those of
the form $\Box A$, which classify computations that produce values of
type $A$ without any access to delayed computations.  The introduction
rule for $\Box$ constructs a delayed computation $\rbox (t)$ that can
be evaluated at any time in the future. This requires $t$ to be typed
in a stable context, and so the hypothesis of the typing rule removes
all ticks and all variables not of stable type from the context. 
The $\Box$ modality has a counit and a
comultiplication $\Box \to \Box\Box$.
%
%
Note that $\inp$ and $\buf$ are stable in the sense that 
$\hastype{}{\rbox (\inp)}{\Box(\DelayE A)}{}$ for any
$\kappa :_\capa A \in \Delta$ where $\capa\in \{\cpush, \cbp\}$ and
$\hastype{}{\rbox(\buf)}{\Box A}$ for any $\kappa :_\capa A \in \Delta$ where
$\capa\in \{\cbuf, \cbp\}$.  


Async RaTT is a terminating calculus in the sense that each step of 
computation terminates. It does, however, still allow recursive definitions 
through a fixed point operator, whose type ensures that recursive calls are
only done in later time steps. More precisely, the recursion variable 
$x$ in $\fix \,x.t$ has type $\DelayA A$, which means that the recursive 
definition can be unfolded to produce a term of type $A$ any time in the 
future, but not now. This is ensured through the elimination rule for $\DelayA$ 
which allows it to be advanced using a tick on any clock typable in the current
context. Since fixed points can be called recursively at any time in the future, 
these must be stable, and so $t$ is required to be typable in a stable context.

The types $\Fix\,\alpha. A$ are guarded recursive types that unfold to
$A[\DelayE(\Fix\,\alpha.A)/\alpha]$ via the terms $\into$ and
$\out$. The most important of these types is $\Sig\, A$ defined as
$\Fix\,\alpha . (A \times \alpha)$, which unfolds to
$A \times \DelayE(\Sig\,A)$. A signal consists of a current value and
a delayed tail, which at some time in the future may return a new
signal. 
Any push channel $\kappa :_\capa A \in \Delta$, where $\capa\in \{\cpush, \cbp\}$, 
induces a stable signal:
\[
 \rbox\left(\fix \,x . \delay_{\cl{\inp}}\left(\into\,\pair{\adv(\inp)}{\adv (x)}\right)\right) : \Box(\DelayE(\Sig\, A))
\]
where the recursion variable $x$ has type $\DelayA(\DelayE A)$. 
These signals, of course, operate on a fixed clock $\{\kappa\}$, 
but in general, the clock associated with the tail of a signal may change from one step 
to the next, which we shall see examples of in \autoref{sec:examples}.

Besides all these constructions, Async RaTT also has a number of standard constructions 
from functional programming: sum types, product types, natural numbers and function
types. The typing rules for these are completely standard, with the exception
that function types can only be constructed in contexts with no ticks. Similar 
restrictions are known from other calculi in the RaTT 
family~\cite{bahr2019simply,bahr2021diamonds}, and are 
necessary for the results of \autoref{sec:oper-semant}. The aforementioned program
transformation by \citet{rattusJFP} also removes this restriction.
Note that function types are not stable, since time-dependent 
references can be stored in closures. 

%


\section{Programming in Async RaTT}
\label{sec:examples}

In this section, we demonstrate the expressiveness of Async RaTT with
a number of examples. To this end, we assume a surface language that
extends Async RaTT with syntactic sugar for pattern matching,
recursion, and top-level definitions. These can be easily elaborated
into the Async RaTT calculus as described in
\autoref{sec:elaboration}. 

\subsection{Simple Signal Combinators}
\label{sec:signal-combinators}

We start by implementing a small set of simple combinators to
manipulate signals, i.e., elements of the guarded recursive type
$\Sig\, A$ defined as $\Fix\,\alpha . (A \times \alpha)$. For
readability we use the shorthand $s \cons t$ for $\into\pair s t$,
such that, given $s : A$ and $t : \DelayE(\Sig\,A)$, we have that
$s \cons t : \Sig\,A$.

We start with perhaps the simplest signal combinator:
\mapCode
The \varid{map} combinator takes a stable function $f$ and applies it
pointwise to a given signal. The fact that $f$ is of type
$\Box (A \to B)$ rather than just $A \to B$ is crucial: Since
$A \to B$ is not a stable type, $f$ would otherwise not be in scope
under the $\delay$, where we need $f$ for the recursive call. It also
has an intuitive justification: The function will be applied to values
of the input signal arbitrarily far into the future, but a closure of
type $A \to B$ may contain references to delayed computations that may
have been garbage collected in the future.

The \varid{map} combinator is stateless in the sense that the current
value of the output signal only depends on the current value of the
input signal. We can generalise this combinator to \varid{scan}, which
produces an output signal that in addition may depend on the previous
value of the output signal:
\scanCode
Every time the input signal updates, the output signal produces a new
value based on the current value of the input signal and the previous
value of the output signal. Since the previous value of the output
signal is accessed, $B$ must be a stable type. We use the
$\Rightarrow$ notation to delineate such constraints from the type signature.

For example, we can use
\varid{scan} to produce the sum of an input signal of numbers:
\sumCode

Often we only have access to a \emph{delayed} signal. For instance,
for each push channel $\kappa :_\capa A \in \Delta$,
$\capa \in \set{\cpush,\cbp}$ we have the signal
\sigAwaitCode
For example, we might have the push-only channels
$\mathsf{mouseClick} :_\cpush 1$ or
$\mathsf{keyPress} :_\cpush \mathsf{KeyCode}$ available. We can derive
a version of \varid{scan} for such signals:
\scanAwaitCode
A simple use case of \varid{scanAwait} is a combinator that counts the
updates of a given delayed signal, e.g., the number of key presses:
\counterCode

Finally, we have the most simple combinator that simply produces a
constant signal:
\constCode
In isolation this combinator may appear to be of little use. Its
utility becomes apparent once we also have the switching
combinators introduced in the next section.

\subsection{Concurrent Signal Combinators}
\label{sec:synchr-comb}

The combinators we looked at so far only consumed a single signal, and
thus had no need to account for the concurrent behaviour of two or
more clocks. For example, we may have two input
signals produced by two redundant sensors that independently provide
a reading we are interested in. To combine these two signals, we can
interleave them using the following combinator:
\interleaveCode
In this and subsequent definitions, we use the shorthands
\varid{Left}, \varid{Right}, and \varid{Both}, in the expected
way. For example, $\varid{Left}\,s\,t$ is short for
$\interm_1{(\interm_1{\pair{s}{t}})}$, i.e., the case that the left
clock ticked first.  The \varid{interleave} combinator uses $\select$
in order to wait until at least one of the input signals ticks, and
then updates the output signal accordingly. In case that both signals
tick simultaneously, the provided merging function $f$ is applied. For
example, $f$ could just always use the value of the first signal or
take the average. Note that the produced signal combines the clocks of
the input signals, i.e., it ticks whenever either of the input signals
ticks.

We might also be interested in the values of both input signals
simultaneously, in which case we would use \varid{zip}:
\zipCode
Similarly to \varid{interleave}, the output signal produced by
\varid{zip} ticks whenever either of the input signals does. However,
note that in the \varid{Left} and \varid{Right} cases, we copy the
previously observed value from the signal that did not tick into the
future. Hence, we need both types, $A$ and $B$, to be stable.

Finally, we consider the switching of signals. We wish to produce a
signal that behaves initially like a given input signal, but switches to
a different signal as soon as some event happens. This idea is
implemented in the \varid{switch} function:
\switchCode
The event that represents the future change of the signal is
represented as a delayed signal, and as soon as this delayed signal
ticks, as in the \varid{Right} and \varid{Both} cases, it takes
over. With the help of \varid{switch} we can construct dynamic
dataflow graphs since we replace a given signal with an entirely new
signal, which may depend on different input channels and intermediate
signals compared to the original signal. 

We will demonstrate an example of this dynamic behaviour in the next
section. In preparation for that we devise a variant of
$\varid{switch}$, where the new signal depends on the value of the
previous signal:
\switchfCode
Instead of a new signal, this combinator waits for a \emph{function}
that produces the new signal, and we feed this function the last value
of the first signal.

\subsection{A Simple GUI Example}
\label{sec:simple-gui-example}

To demonstrate how to use our signal combinators, we consider a very
simple example of a GUI application: Our goal is to write a reactive
program with two output channels that describe the contents of two
text fields. To this end, the two output channels are given the type
$\Sig\,\Nat$. The number displayed in text fields should be
incremented each time the user clicks a button, which is available as
an input channel $\mathsf{up} :_\cpush 1\in \Delta$. However, there is only
one `up' button and the user can change which text field should be
changed by the `up' button using a `toggle' button, which is available
as an input channel $\mathsf{toggle} :_\cpush 1\in \Delta$.

That means, the contents of the first text field can be described by
the \varid{count} combinator, but then switches to the signal
described by the \varid{const} combinator when `toggle' is
pressed. The behaviour of the other text field is reversed: first
\varid{const}, then \varid{count}. This continuous toggling between
behaviours can be concisely described by the following combinator:
\toggleSigCode
The first argument provides the events that determine when to toggle
between the two behaviours, which in turn are given as the next two
arguments. In the implementation we use the notation $s;t$ as a
shorthand for $\letterm{\unit}{s}{t}$. The \varid{toggleSig}
combinator uses \varid{switchf} to start with the first signal
provided by $f$, but then switches to $g$ as soon as the toggle
\varid{tog} ticks by using a recursive call that swaps the order of the
two arguments $f$ and $g$.

The output channels that describe the two text fields can now be
implemented by providing the appropriate input signals to \varid{toggleSig}:
\guiExampleCode
Note that the dataflow graph changes during the execution of the
program and how that change is reflected in the clock associated with
the output channels: the output channel for the first text field first
has the clock $\set{\sym{up},\sym{toggle}}$ as it must both count the
number of times the `up' button is clicked and change its behaviour in
reaction to the `toggle' button being clicked. Once the `toggle'
button has been clicked, the clock for output channel for the text
field changes to $\set{\sym{toggle}}$ as it now ignores the `up'
button. We will examine the run-time behaviour of this example in more
detail in \autoref{sec:example}.

\subsection{Integral and Derivative}

\begin{figure}
  \integral
  \derivative
  \caption{Integral and derivative signal combinators.}
  \label{fig:integral-derivative}
\end{figure}

Buffered input channels can be used to represent input signals that
change at discrete points in time, but whose current value can be
accessed at any time. For example, given a buffered push channel
$\kappa:_\cbp A \in \Delta$, we can construct the following signal
(using $\varid{sigAwait}_{\kappa}$ from \autoref{sec:signal-combinators}):
\sigCode
To illustrate what we can do with such input signals, we assume that
Async RaTT has a stable type $\Float$ together with typical operations
on floating-point numbers. Figure~\ref{fig:integral-derivative} gives
the definition of two signal combinators that each take a
floating-point-valued signal and produce the integral and the
derivative of that signal. To this end, we assume a buffered push
channel $\mathsf{sample} :_{\cbp} \Float \in \Delta$ that produces
a new floating-point number $s$ at some fixed interval (e.g., 10 times
per second). This number $s$ is the number of seconds since the last
update on the channel, e.g., $s = 0.1$ if $\mathsf{sample}$ ticks $10$
times per second.

The \varid{integral} combinator produces the integral of a given
signal starting from a given constant that is provided as the first
argument. Its implementation uses a simple approximation that samples
the value of the underlying signal each time the $\mathsf{sample}$
channel produces a value and adds the area of the rectangle formed by
the value of the signal and the time that has passed since the last
sampling.

The first equation of the definition is an optimisation and could be
omitted. It says that if the current value of the underlying signal is
$0$, we simply wait until the underlying signal is updated, since the
value of the integral won't change until the underlying signal has a
non-zero value. Hence, we don't have to sample every time the
$\mathsf{sample}$ channel ticks.

Similarly to $\varid{integral}$, we can implement a function
$\varid{derivative}$ that, given a floating-point-valued
signal, produces its derivative. Like the $\varid{integral}$ function,
also $\varid{derivative}$ samples the underlying signal every time
$\mathsf{sample}$ ticks. To do so it uses the auxiliary function
$\varid{der}$, which takes two additional arguments: the current value
of the derivative and the value of the underlying signal at the time
of the most recent input from of the $\mathsf{sample}$
channel. Similarly to \varid{integral}, the first line of
$\varid{der}$ performs an optimisation: If the computed value of the
derivative is $0$, the sampling will pause until the underlying signal
is updated. As soon as it does, we pretend that $\mathsf{sample}$
ticked to provide a timely update of the derivative.

These two combinators can be easily generalised from floating-point
values to any vector space. This can then be used to describe complex
behaviours in reaction to multidimensional sensor data.

\subsection{Elaboration of Surface Syntax into Core Calculus}
\label{sec:elaboration}

To illustrate how the surface language elaborates into the Async RaTT
core calculus, reconsider the definition of \varid{map}
\mapCode
which elaborates to the following term in plain Async RaTT:
\[
  \begin{aligned}
    \emph{map} = \fix\,r.\lambda f. \lambda s .
    &\letterm{x}{\pi_1(\out\,s)}
    {\sym{let}\, \varid{xs} = \pi_2(\out\,s)}\\
    &\sym{in} \,\into{\pair{unbox\, f\,
            x}{\delay_{\cl{\varid{xs}}} (\adv\,r\, f\, (\adv\, \varid{xs}))}}
  \end{aligned}
\]
Recall that $s \cons t$ is a shorthand for $\into\pair s t$. Pattern
matching is translated into the corresponding elimination forms,
$\out$ for recursive types, $\pi_i$ for product types, and $\case$ for
sum types. The recursion syntax -- \varid{map} occurs in the body of
its definition -- is translated to a fixed point $\fix \, r . t$ so
that the recursive occurrence of \varid{map} is replaced by
$\adv\,r$. Hence, recursive calls must always occur in the scope of a
$\tick$, which is the case in the definition of \varid{map} as it
appears in the scope of a $\delay$. Moreover, we elide the subscript
$\cl{\varid{xs}}$ of $\delay$ since it can be uniquely inferred from
the fact that we have the term $\adv\,\varid{xs}$ in the scope of the
$\delay$.

In addition, we make use of top-level definitions like \varid{map} and
\varid{scan}, which may be used in any context later on. For example,
\varid{scan} is used in the definition of \varid{scanAwait} in the
scope of a $\tick$. We can think these top-level definitions to be
implicitly boxed when defined and unboxed when used later on. That is,
these definitions are translated as follows to the core calculus:
\begin{align*}
  &\letterm{\varid{scan}}{\rbox (\dots)}{}\\
  &\letterm{\varid{scanAwait}}{\rbox(\dots(\unbox\, scan)\dots))}{}\\
  &\dots
\end{align*}

\section{Operational Semantics and Operational Guarantees}
\label{sec:oper-semant}
We describe the operational semantics of Async RaTT in two stages: We
begin in section~\ref{sec:evaluation-semantics} with the
\emph{evaluation semantics} that describes how Async RaTT terms are
evaluated at a particular point in time. Among other things, the
evaluation semantics describes the computation that must happen to
make updates in reaction to the arrival of new input on a push
channel. We then describe in section~\ref{sec:step-semantics} the
\emph{reactive semantics} that captures the dynamic behaviour of Async
RaTT programs over time. The reactive semantics is a machine that
waits for new input to arrive, and then computes new values for output
channels that depend on the newly arrived input. For the latter, the
reactive semantics invokes the evaluation semantics to perform the
necessary updating computations.

Finally, after demonstrating the operational semantics on an example
in section~\ref{sec:example}, we conclude the discussion of the
operational semantics in section~\ref{sec:main-results} with a precise
account of our main technical results about the properties of the
operational semantics: productivity, causality, signal independence,
and the absence of implicit space leaks. To prove the latter, the
evaluation semantics uses a store in which both external inputs and
delayed computations are stored. Delayed computations are garbage
collected as soon as the data on which they depend has arrived. In
this fashion, Async RaTT avoids implicit space leaks by construction,
provided we can prove that the operational semantics never gets stuck.


\subsection{Evaluation Semantics}
\label{sec:evaluation-semantics}

\begin{figure}[t]
  \begin{mathpar}
    \small
  \inferrule*%
  {~}%
  {\heval{v}{\sigma}{v}{\sigma}}%
  \and%
  \inferrule*%
  {\heval {t} {\sigma} {v} {\sigma'}\\
    \heval {t'} {\sigma'} {v'} {\sigma''}}%
  {\heval {\pair{t}{t'}} {\sigma} {\pair{v}{v'}} {\sigma''}}%
  \and%
  \inferrule*%
  {\heval {t} {\sigma} {\pair{v_1}{v_2}} {\sigma'} \\ i \in \{1,2\}}%
  {\heval {\pi_i(t)} {\sigma} {v_i} {\sigma'}}%
  \and%
  \inferrule*%
  {\heval t {\sigma} v {\sigma'} \\ i \in \{1,2\}}%
  {\heval {\interm_i(t)} {\sigma} {\interm_i(v)} {\sigma'}}%
  \and%
  \inferrule*%
  {\heval {t} {\sigma} {\interm_i(v)} {\sigma'}\\
    \heval {t_i[v/x]} {\sigma'} {v_i} {\sigma''} \\ i \in \{1,2\}}%
  {\heval {\caseterm{t}{x}{t_1}{x}{t_2}} {\sigma} {v_i} {\sigma''}}%
  \and%
  \inferrule*%
  {\heval{t}{\sigma}{\lambda x.s}{\sigma'} \\
    \heval{t'}{\sigma'}{v}{\sigma''}\\
    \heval {s[v/x]}{\sigma''}{v'}{\sigma'''}}%
  {\heval{t\,t'}{\sigma}{v'}{\sigma'''}}%
  \and%
  \inferrule*%
  {\heval{s}{\sigma}{v}{\sigma'} \\
    \heval {t[v/x]}{\sigma'}{w}{\sigma''}}%
  {\heval{\letterm x s t}{\sigma}{w}{\sigma''}}%
    \and%
  \inferrule*%
  {\kappa \in \dom\iota}%
  {\heval {\buf[\kappa]}{\sigma}{\iota(\kappa)}{\sigma}}%
  \and%
  \inferrule*%
  {l = \allocate[\cleval{\theta}]{\sigma}}%
  {\heval{\delay_{\theta}\,t}{\sigma}{l}{(\sigma,l \mapsto t)}}%
  \and%
  \inferrule*%
  {l = \allocate[\emptyset]{\sigma}}%
  {\heval{\never}{\sigma}{l}{\sigma}}%
  \and%
  \inferrule*%
  {~}%
  {\heval
    {\adv\,\inp[\kappa]}{\eta_N\stick{\kappa}{v}\eta_L}{v}{\eta_N\stick{\kappa}{v}\eta_L}}%
  \and%
  \inferrule*%
  {\heval {\eta_N(l)}{\eta_N\stick{\kappa}{v}\eta_L} {w}{\sigma}}%
  {\heval {\adv\,l}{\eta_N\stick{\kappa}{v}\eta_L}{w}{\sigma}}%
  \and%
  \inferrule*%
  {\kappa\in \cleval{\cl{v_i}}\setminus\cleval{\cl{v_{3-i}}} \\ \heval
    {\adv\, v_i}{\eta_N\stick{\kappa}{w}\eta_L} {u_i}{\sigma} \\
  u_{3-i} = v_{3-i}}%
  {\heval
    {\select\,v_1\,v_2}{\eta_N\stick{\kappa}{w}\eta_L}{\interm_1(\interm_i\pair
      {u_1} {u_2})}{\sigma}}%
  \and%
  \inferrule*%
  {\kappa\in \cleval{\cl{v_1}}\cap\cleval{\cl{v_2}} \\ \heval
    {\adv\,v_1}{\eta_N\stick{\kappa}{w}\eta_L} {u_1}{\sigma}
  \\ \heval {\adv\,v_2}{\sigma} {u_2}{\sigma'}}%
  {\heval
    {\select\,v_1\,v_2}{\eta_N\stick{\kappa}{w}\eta_L}{\interm_2\pair{u_1}{u_2}}{\sigma'}}%
  \and%
  \inferrule*%
  {\heval{t}{\sigma}{v}{\sigma'}}%
  {\heval{\suc \, t}{\sigma}{\suc \, v}{\sigma'}}%
  \and%
  \inferrule*%
  {\heval{n}{\sigma}{\zero}{\sigma'}
    \\ \heval{s}{\sigma'}{v}{\sigma''}}%
  {\heval{\recN{s}{x}{y}{t}{n}}{\sigma}{v}{\sigma''}}%
  \and%
  \inferrule*%
  {\heval{n}{\sigma}{\suc \, v}{\sigma'}
    \\ \heval{\recN{s}{x}{y}{t}{v}}{\sigma'}{v'}{\sigma''}
    \\ \heval{t[v/x,v'/y]}{\sigma''}{w}{\sigma'''}}%
  {\heval{\recN{s}{x}{y}{t}{n}}{\sigma}{w}{\sigma'''}}%
    \and%
  \inferrule*%
  {\heval
    {t[\dfix\,x.t/x]}{\sigma} {v}{\sigma'}}%
  {\heval {\adv\,(\dfix\,x. t)}{\sigma}{v}{\sigma'}}%
  \and%
    \inferrule*%
  {\heval
    {t[\dfix\,x.t/x]}{\sigma} {v}{\sigma'}}%
  {\heval {\fix\,x. t}{\sigma}{v}{\sigma'}}%
    \and%
  \inferrule*%
  {\heval {t}{\sigma} {\rbox\, t'}{\sigma'} \\ \heval
    {t'}{\sigma'} {v}{\sigma''}}%
  {\heval {\unbox\,t}{\sigma}{v}{\sigma''}}%
  \and
    \inferrule*
  {\heval{t}{\sigma}{v}{\sigma'}}
  {\heval{\into\,t}{\sigma}{\into\,v}{\sigma'}}
  \and
  \inferrule*
  {\heval{t}{\sigma}{\into\,v}{\sigma'}}
  {\heval{\out\,t}{\sigma}{v}{\sigma'}}
\end{mathpar}
  
  \caption{Operational semantics.}
  \label{fig:machine}
\end{figure}

Figure~\ref{fig:machine} defines the evaluation semantics as a
deterministic big-step operational semantics. We write
$\heval t \sigma v \tau$ to denote that when given a term $t$, a
\emph{store} $\sigma$, and an \emph{input buffer} $\iota$, the machine
computes a value $v$ and a new store $\tau$. During the computation,
the machine may defer computations into the future by storing
unevaluated terms in the store $\sigma$ to be retrieved and evaluated
later. Conversely, the machine may also retrieve terms whose
evaluation have been deferred at an earlier time and evaluate them
now. In addition, the machine may read the new value of the most
recently updated push channel from the store $\sigma$ and read the
current value of any buffered channel from the input buffer $\iota$.

To facilitate the delay of computations, the syntax of the language
features heap locations $l$, which are not typable in the calculus but
may be introduced by the machine during evaluation. A heap location
represents a delayed computation that can be resumed once a particular
clock has ticked, which indicates that the data the delayed
computation is waiting for has arrived. To this end, each heap
location $l$ is associated with a clock, denoted $\cl{l}$. As soon as
the clock $\cl{l}$ ticks, the delayed computation represented by $l$
can be resumed by retrieving the unevaluated term stored at heap
location $l$ and evaluating it. We write $\locs$ for the set of all
heap locations and assume that for each clock $\Theta$, there are
countably infinitely many locations $l$ with $\cl l = \Theta$. A clock
$\Theta$ is a finite set of push channels drawn from $\dom\Delta$, and
it ticks any time any of its channels $\kappa \in \Theta$ is
updated. For example, assuming an input context $\Delta$ for a GUI,
the clock $\set{\sym{keyPressed}, \sym{mouseCoord}}$ ticks whenever
the user presses a key or moves the mouse. Note that we are now more
precise in distinguishing \emph{clock expressions}, typically denoted
$\theta$, and \emph{clocks}, typically denoted $\Theta$. A closed
clock expression $\theta$ evaluates to a clock $\cleval\theta$ as
follows:
\begin{align*}
  \cleval{\cl l} = \cl l \qquad%
  \cleval{\cl{\inp[\kappa]}} = \set\kappa \qquad%
  \cleval{\theta \sqcup \theta'} =  \cleval{\theta} \cup \cleval{\theta'}
\end{align*}

Delayed computations reside in a heap, which is simply a finite
mapping $\eta$ from heap locations to terms. Of particular interest
are heaps $\eta$ whose locations, denoted $\dom{\eta}$, each have a
clock that contains a given input channel $\kappa$:
\[
\heaps[\kappa] = \setcom{\eta\in \heaps}{\forall l \in \dom{\eta} .\ \kappa \in
  \cl{l}}
\]
It is safe to evaluate terms stored in a heap $\eta \in \heaps[\kappa]$ as
soon as a new value on the input channel $\kappa$ has arrived. This
intuition is reflected in the representation of stores $\sigma$, which
can be in one of two forms: a single-heap store $\eta_L$ or a two-heap
store $\eta_N \stick{\kappa}{v} \eta_L$ with
$\eta_N \in \heaps[\kappa]$. We typically refer to $\eta_L$ as the
\emph{later} heap, which is used to store delayed computations for
later, and to $\eta_N$ as the \emph{now} heap, whose stored terms are
safe to be evaluated now. The $\stick{\kappa}{v}$ component of a
two-heap store indicates that the input channel $\kappa$ has been
updated to the new value $v$. The machine can thus safely resume
computations from $\eta_N$ since the data that the delayed
computations in $\eta_N$ were waiting for has arrived.

Let's first consider the semantics for $\delay$: To allocate fresh
locations in the store, we assume a function $\allocate{\cdot}$,
which, if given a store $\eta_L$ or $\eta_N \stick{\kappa}{v} \eta_L$,
produces a location $l \nin \dom{\eta_L}$ with $\cl{l} = \Theta$. This
results in a store $\eta_L,l\mapsto t$ or
$\eta_N \stick{\kappa}{v} \eta_L,l \mapsto t$, respectively, where
$\eta_L,l\mapsto t$ denotes the heap $\eta_L$ extended with the
mapping $l\mapsto t$.

Conversely, $\adv\,l$ retrieves a previously delayed computation. The
typing discipline ensures that $\adv\,l$ will only be evaluated in the
context of a store of the form $\eta_N\stick{\kappa}{v}\eta_L$ with
$l \in \dom{\eta_N}$ and therefore also $\kappa \in \cl{l}$. In
addition, $\adv$ may be applied to $\inp[\kappa]$ which simply looks
up the new value $v$ from the channel $\kappa$.

The $\select$ combinator allows us to interact with two delayed
computations simultaneously. Its semantics checks for the three
possible contingencies, namely which non-empty subset of the two
delayed computations has been triggered. Each of the two argument
values $v_1$ or $v_2$ is either a heap location or $\inp$ and thus the
machine can simply check whether the current input channel $\kappa$ is
in the clocks associated with $v_1$, $v_2$, or both. Depending on the
outcome, the machine advances the corresponding value(s).

Finally, the fixed point combinator $\fix$ is evaluated with the help
of the combinator $\dfix$, which similarly to heap locations is not
typable in the calculus but is introduced by the machine. Intuitively
speaking, we can think of a value of the form $\dfix\,x.t$ as
shorthand for $\Lambda \theta . (\delay_\theta (\fix\,x.t))$. That is,
$\dfix\,x.t$ is a thunk that, when given a clock $\theta$, produces a
delayed computation on $\theta$, which in turn evaluates a fixed point
once $\theta$ ticks. The action of the $\adv$ combinator for $\DelayA$
can thus also be interpreted as first providing the clock $\theta$ and
then advancing the delayed computation $\delay_\theta (\fix\,x.t)$,
which means evaluating $\fix\,x.t$.

\subsection{Reactive Semantics}
\label{sec:step-semantics}
%

An Async RaTT program interacts with its environment by
receiving input from a set of input channels and in return sends
output to a set of output channels. The input context $\Delta$
describes the available input channels. In addition, we also have an
output context $\Gammaout$, that only contains variables
$x : A$, where $A$ is a value type. We refer to the variables in
$\Gammaout$ as output channels. Taken together, we call the pair
consisting of $\Delta$ and $\Gammaout$ a \emph{reactive interface},
written $\Delta \Rightarrow \Gammaout$.

Given an output interface $\Gammaout = x_1 : A_1,\dots,x_n: A_n$, we
define the type $\Prod\Gammaout$ as the product of all types in
$\Gammaout$, i.e.,
$\Prod\Gammaout = \Sig\,A_1 \times \dots \times \Sig\,A_n$. The
$n$-ary product type used here can be encoded using the binary product
type and the unit type in the standard way. An Async RaTT term $t$ is
said to be a reactive program implementing the reactive interface
$\Delta \Rightarrow \Gammaout$, denoted
$t : \Delta \Rightarrow \Gammaout$, if
$\hastype[\Delta]{}{t}{\Prod{\Gammaout}}$.

\begin{figure}
  \begin{mathpar}
    \inferrule*[left=init]%
    {\heval{\tuple t }{\emptyset}{\tuple{ v_1\cons
          l_1,\dots,v_m \cons l_m}}{\eta}}%
    {\stateI{t}\forward{x_1\mapsto
        v_1,\dots,x_m\mapsto v_m} \stateS{x_1\mapsto
        l_1,\dots,x_m\mapsto l_m}{\eta}}%
    \and%
    \inferrule*[left=input]%
    {\iota' = \iota[\kappa \mapsto v] \text{ if }\kappa\in\dom\iota
      \text{ otherwise } \iota' = \iota }
    {\stateS{N}{\eta}\forward{\kappa \mapsto v}
      \stateS[\iota']{N}{\clin{\eta}\stick{\kappa}{v}\clnin{\eta}}}%
    \and%
    \inferrule*[left=output-end]%
    {~}%
    {\stateS{\cdot}{\eta_N\stick{\kappa}{v}\eta_L}\forward{\cdot}
      \stateS{\cdot}{\eta_L}}%
    \and%
    \inferrule*[left=output-skip]%
    {\kappa\nin\cl l\\\stateS{N}{\eta_N\stick{\kappa}{v}\eta_L}\forward{O}
      \stateS{N'}{\eta}}%
    {\stateS{x\mapsto l,N}{\eta_N\stick{\kappa}{v}\eta_L}\forward{O}
      \stateS{x\mapsto l,N'}{\eta}}%
    \and%
    \inferrule*[left=output-compute]%
    {\kappa\in\cl l\\
      \heval{\adv\,l}{\eta_N\stick{\kappa}{v}\eta_L}{v'\cons l'}{\sigma}\\
      \stateS{N}{\sigma}\forward{O}
      \stateS{N'}{\eta}}%
    {\stateS{x\mapsto l,N}{\eta_N\stick{\kappa}{v}\eta_L}\forward{x\mapsto
        v', O}
      \stateS{x\mapsto l',N'}{\eta}}%
\end{mathpar}
  \caption{Reactive semantics.}
  \label{fig:reactiveSemantics}
\end{figure}

The operational semantics of a reactive program is described by the
machine in Figure~\ref{fig:reactiveSemantics}. The state of the
machine can be of two different forms: Initially, the machine is in a
state of the form $\stateI t$, where $t: \Delta \Rightarrow \Gammaout$
is the reactive program and $\iota$ is the initial input buffer, which
contains the initial values of all buffered input
channels. Subsequently, the machine state is a pair
$\stateS{N}{\sigma}$, where $N$ is a sequence of the form
$x_1 \mapsto l_1,\dots,x_n \mapsto l_n$ that maps output channels
$x_i \in \dom{\Gammaout}$ to heap locations. That is, $N$ records for
each output channel the location of the delayed computation that will
produce the next value of the output channel as soon as it needs
updating.

The machine can make three kinds of transitions:
\begin{align*}
  &\text{an initialisation
    transition}
  &\stateI{t} &\forward{O} \stateS{N}{\eta}\\
  &\text{an input transition}&\stateS{N}{\eta} &\forward{\kappa\mapsto v}
  \stateS[\iota']{N}{\eta_N\stick{\kappa}{v}\eta_L}\\
  &\text{an output transition}&\hspace{-3cm}
  \stateS{N}{\eta_N\stick{\kappa}{v}\eta_L} &\forward{O}
  \stateS{N'}{\eta_L}
\end{align*}
where $O$ is a sequence $x_1\mapsto v_1,\dots, x_n\mapsto v_n$ that
maps output channels to values.  After the initial transition, which
initialises the values of all output channels, the machine alternates
between input transitions, each of which updates the value of an input
channel and possibly the input buffer (if the new input is on a
buffered channel), and output transitions, each of which provides new
values for all those output channels triggered by the immediately preceding
input transition.

The initialisation transition evaluates the reactive program $t$ in
the context of the initial input buffer $\iota$ and thereby produces a
tuple $\tuple{ v_1\cons l_1,\dots,v_m \cons l_m}$ whose components
$v_i\cons l_i$ correspond to the output channels
$x_i : A_i \in \Gammaout$. Each $v_i$ is the initial value of
the output channel $x_i$ and each $l_i$ points to a delayed
computation in the heap $\eta$ that computes future values of $x_i$.

An input transition receives an updated value $v$ on the input channel
$\kappa$ and reacts by updating the input buffer (if it already had a
value for $\kappa$) and transitioning the store $\eta$ to the new
store $\clin{\eta}\stick{\kappa}{v}\clnin{\eta}$. This splits the heap
$\eta$ into a part that contains $\kappa$ and a part that does not:
\[
  \clin{\eta}(l) = \eta(l) \quad \text{if } \kappa \in \cl{l} \hspace{2cm}
  \clnin{\eta}(l) = \eta(l) \quad \text{if } \kappa \nin \cl{l}
\]
That is, in the subsequent output transition, the machine can read
from $\clin{\eta}$, i.e., exactly those heap locations from $\eta$
that were waiting for input from $\kappa$, and access the new value
$v$ from $\kappa$.

Finally, the output transition checks for each element $x \mapsto l$ in
$N$, whether it should be advanced because it depends on $\kappa$
(\textsc{output-compute}) or should remain untouched because it does
not depend on $\kappa$ (\textsc{output-skip}). Only in the
\textsc{output-compute} case, a new output value for $x$ is produced.
In the end, the output transition performs the desired garbage
collection that deletes both the now heap $\eta_N$ and the input value
$v$ (\textsc{output-end}). This also means that the updates performed
by \textsc{output-compute}, are not only possible (because the
required data arrived), but also necessary (because both the input
data and the delayed computations they depend on will be gone after
this output transition of the machine).

\subsection{Example}
\label{sec:example}

To see the operational semantics in action, we revisit the simple
GUI program from section~\ref{sec:simple-gui-example} and run it on
the machine. To this end, we first elaborate the definition of
$\varid{toggleSig}$ into an explicit fixed point term of the core
calculus as described in section~\ref{sec:elaboration}:
\begin{align*}
  \varid{toggleSig} =\ &\fix\, r . \lambda \varid{tog}.\lambda f.\lambda g
  .\lambda
  x.\\&\letterm{\varid{tick}}{\unbox\,\varid{tog}}{\varid{switchf}\,(\unbox\,f\,x)(\delay_{\cl{\varid{tick}}}\,(adv\,\varid{tick};\adv\,r
   \,\varid{tog}\,g\,f))}
\end{align*}
During the execution, the machine turns fixed points like
\varid{toggleSig} into delayed fixed points that use $\dfix$ instead
of $\fix$. We write $\varid{toggleSig}'$ for this delayed fixed point,
i.e., $\varid{toggleSig}'$ is obtained from $\varid{toggleSig}$ by
replacing $\fix$ with $\dfix$. We will use the same notational
convention for other fixed point definitions and write
$\varid{sigAwait}'_\kappa$ and $\varid{scan}'$ for the $\dfix$
versions of $\varid{sigAwait}_\kappa$ and $\varid{scan}$ from
section~\ref{sec:signal-combinators}.

We consider the program
$\varid{field1} : \Delta \Rightarrow \Gammaout$ with
$\Delta=\set{\mathsf{up} :_\cpush \Unit,\mathsf{toggle} :_\cpush
  \Unit}$, $\Gammaout = x : \Nat$, and
\begin{align*}
  \varid{field1} = \varid{toggleSig}\,t\,s_1\,s_2\, 0
  \quad\text{where }\
  \begin{aligned}[t]
    t &= \rbox\,\inp[\mathsf{toggle}]\\ s_1 &=
    \rbox\,(\varid{count}\,\varid{sigAwait}_{\mathsf{up}})\\
    s_2 &= \rbox\, \varid{const}
  \end{aligned}
\end{align*}
That is, this program describes the behaviour of the text field that
initially is in focus and thus reacts to the `up' button.

For better clarity of the transition steps of the machine, we write
the machine's store as just the list of its heap locations, and write
the contents of the locations along with their clocks separately
underneath. The first step of the machine performs the initialisation
that provides the initial value of the output signal:
\begin{gather*}
  \stateI[\emptyset]{\varid{field1}} \forward{x \mapsto 0}
  \stateS[\emptyset]{x\mapsto l_1}{l_1,l_2,l_3,l_4}\\\text{where}\quad
  \begin{aligned}[t]
    l_1&\mapsto
    \caset{\select\,l_3\,l_4}\dots&\cl{l_1}&=\set{\mathsf{toggle},\mathsf{up}}\\
    l_2&\mapsto \adv\,\inp[\mathsf{up}]  \cons \adv\,\varid{sigAwait}'&\cl{l_2}&=\set{\mathsf{up}}\\
    l_3&\mapsto \varid{scan}\,(\rbox\,\lambda n.\lambda
    m.m+1)\,0\,(\adv\,l_2)&\cl{l_3}&=\set{\mathsf{up}}\\
    l_4&\mapsto \adv\,\inp[\mathsf{toggle}];\adv\,\varid{toggleSig}'\,
    t\, s_2\, s_1&\cl{l_4}&=\set{\mathsf{toggle}}
  \end{aligned}
\end{gather*}
We can see that the next value for the output channel $x$ is provided
by the delayed computation at location $l_1$, and since
$\cl{l_1}=\set{\mathsf{toggle},\mathsf{up}}$ we know that $x$ will
produce a new value as soon as the user clicks either of the two
buttons. If the user clicks the `up' button, we see the following:
\begin{gather*}
  \stateS[\emptyset]{x\mapsto l_1}{l_1,l_2,l_3,l_4}
  \forward{\mathsf{up} \mapsto \unit} \stateS[\emptyset]{x\mapsto l_1}{l_1,l_2,l_3\stick{\mathsf{up}}{\unit}l_4} \forward{x \mapsto 1}
  \stateS[\emptyset]{x\mapsto l_5}{l_4,l_5,l_6,l_7}\\\text{where}\quad
  \begin{aligned}[t]
    l_5&\mapsto
    \caset{\select\,l_7\,l_4}\dots&\cl{l_5}&=\set{\mathsf{toggle},\mathsf{up}}\\
    l_6&\mapsto \adv\,\inp[\mathsf{up}]  \cons \adv\,\varid{sigAwait}'&\cl{l_6}&=\set{\mathsf{up}}\\
    l_7&\mapsto \adv\,\varid{scan}'\,(\rbox\,\lambda n.\lambda
    m.m+1)\,0\,(\adv\,l_6)&\cl{l_7}&=\set{\mathsf{up}}
  \end{aligned}
\end{gather*}
The heap locations $l_1,l_2,l_3$ are garbage collected and only $l_4$
survives since only the clock of $l_4$ does not contain
$\mathsf{up}$. If the user now clicks the `toggle` button, we see the
following:
\begin{gather*}
  \stateS[\emptyset]{x\mapsto l_5}{l_4,l_5,l_6,l_7}
  \forward{\mathsf{toggle} \mapsto \unit} \stateS[\emptyset]{x\mapsto
    l_5}{l_4,l_5\stick{\mathsf{toggle}}{\unit}l_6,l_7}
  \forward{x \mapsto 1}
  \stateS[\emptyset]{x\mapsto l_8}{l_6,l_7,l_8,l_9}\\\text{where}\quad
  \cl{l_0} =\emptyset\quad
  \begin{aligned}[t]
    l_8&\mapsto
    \caset{\select\,l_0\,l_9}\dots&\cl{l_8}&=\set{\mathsf{toggle}}\\
    l_9&\mapsto \adv\,\inp[\mathsf{toggle}]  \cons \adv\,\varid{sigAwait}'\,t\,s_1\,s_2&\cl{l_9}&=\set{\mathsf{toggle}}
  \end{aligned}
\end{gather*}
The heap location $l_0$ is allocated by $\never$ and thus does not
appear on the heap. Now the output channel $x$ only depends on the
input channel $\mathsf{toggle}$. If the user now repeatedly clicks the
`up' button, no output is produced:
\begin{align*}
  \stateS[\emptyset]{x\mapsto l_8}{l_6,l_7,l_8,l_9}
  &\forward{\mathsf{up} \mapsto \unit} \stateS[\emptyset]{x\mapsto
    l_8}{l_6,l_7\stick{\mathsf{up}}{\unit}l_8,l_9} \forward{\cdot}
  \stateS[\emptyset]{x\mapsto
    l_8}{l_8,l_9}\\
  &\forward{\mathsf{up} \mapsto \unit} \stateS[\emptyset]{x\mapsto
    l_8}{\stick{\mathsf{up}}{\unit}l_8,l_9} \forward{\cdot}
  \stateS[\emptyset]{x\mapsto
    l_8}{l_8,l_9}
\end{align*}

Finally, note that since the input context $\Delta$ contains no
buffered input channels the input buffer remains empty during the
entire run of the program.

\subsection{Main Results}
\label{sec:main-results}

The operational semantics presented above allows us to precisely state
the operational guarantees provided by Async RaTT, namely
productivity, the absence of implicit space leaks, causality, and
signal independence. We address each of them in turn.

\subsubsection{Productivity}
\label{sec:productivity}

Reactive programs $t : \Delta \Rightarrow \Gammaout$ are productive in
the sense that if we feed $t$ with a well-typed initial input buffer
and an infinite sequence of well-typed inputs on its input channels,
then it will produce an infinite sequence of well-typed outputs on its
output channels. Before we can state the productivity property formally, we
need to make precise what we mean by well-typed:
\begin{itemize}
\item An input buffer $\iota$ is well-typed , denoted
  $\vdash \iota : \Delta$, if $\hastype[]{}{\iota(\kappa)}{A}$ for
  each $\kappa$ such that $\kappa :_{\cbuf} A \in \Delta$ or
  $\kappa :_{\cbp} A \in \Delta$.
\item An input value $\kappa \mapsto v$ is well-typed, written
  $\vdash \kappa \mapsto v : \Delta$, if $\kappa :_\capa A \in \Delta$
  and $\hastype[]{}{v}{A}$.
\item A set of output values $O$ is well-typed, written
  $\vdash O : \Gammaout$, if for all $x \mapsto v \in O$, we have that
  $x : A \in \Gammaout$ and $\hastype[]{}{v}{A}$.
\end{itemize}

We can now formally state the productivity property as follows:
\begin{theorem}[productivity]
  \label{thr:productivity}
  Given a reactive program $t : \Delta \Rightarrow \Gammaout$,
  well-typed input values $\vdash \kappa_i\mapsto v_i : \Delta$ for
  all $i \in \nats$, and a well-typed initial input buffer
  $\vdash \iota_0 : \Delta$, there is an infinite transition sequence
\[
  \stateI[\iota_0]{t} \forward{O_0} \stateS[\iota_0]
  {N_0}{\eta_0}\forward{\kappa_0\mapsto v_0}\stateS[\iota_1]{N_0}{\sigma_0}
  \forward{O_1}\stateS[\iota_1]{N_1}{\eta_1}\forward{\kappa_1\mapsto v_1}\dots
\]
with $\vdash O_i : \Gammaout$ for all $i \in \nats$.
\end{theorem}
While a reactive program will always produce a set of output values
$O_{i+1}$ for each incoming input value $\kappa_i \mapsto v_i$, this
set may be empty. This happens if none of
the heap locations in $N_i$ depends on the input $\kappa_i$, i.e., if
$\kappa_i \nin \cl{l}$ for all $x \mapsto l \in N_i$. As we will see
in Proposition~\ref{prop:bufSigIndependence}, this will necessarily be
the case for inputs $\kappa :_{\cbuf} A \in \Delta$ that are
buffered-only. Note that all output channels are initialised in the initialisation
transition. An empty set of output values therefore only means that no
output channels need to be updated. 

\subsubsection{Implicit Space Leaks}
\label{sec:implicit-space-leaks}

The absence of implicit space leaks is a direct consequence of the
productivity property (Theorem~\ref{thr:productivity}). More
precisely, the operational semantics of Async RaTT is formulated in
such a way that after each pair of input/output transitions
\[
  \stateS{N}{\eta} \forward{\kappa\mapsto v}
  \stateS[\iota']{N}{\eta_N\stick{\kappa}{v}\eta_L} \forward{O}
  \stateS[\iota']{N'}{\eta_L}
\]
all heap locations $l$ in $\eta$ that depend on $\kappa$, i.e., those
with $\kappa \in \cl l$, are garbage collected and thus do not appear
in $\eta_L$. That is, a delayed computation at location $l$ is only
kept in memory until its clock $\cl l$ ticks. By
Theorem~\ref{thr:productivity}, this aggressive garbage collection
strategy is safe: The machine never gets stuck attempting to
dereference a garbage collected heap location.

\subsubsection{Causality}
\label{sec:causality}

In the following we refer to the transition sequences for a reactive
program $t$ obtained by Theorem~\ref{thr:productivity} simply as
well-typed transition sequences for $t$.

A reactive program $t$ is causal, if for any of its well-typed
transition sequences
\begin{equation}  \label{eq:causality:trans:seq}
  \stateI[\iota_0]{t} \forward{O_0} \stateS[\iota_0]
  {N_0}{\eta_0}\forward{\kappa_0\mapsto v_0}\stateS[\iota_1]{N_0}{\sigma_0}
  \forward{O_1}\stateS[\iota_1]{N_1}{\eta_1}\forward{\kappa_1\mapsto v_1}\dots
\end{equation}
each set of output values $O_n$ only depends on the initial input
buffer $\iota_0$ and previously received input values
$\kappa_i \mapsto v_i$ with $i < n$. To see that this is always the
case, we first note that the operational semantics is deterministic in
the following sense:
%
\begin{lemma}[deterministic semantics]
  \label{lem:deterministic}
  ~
  \begin{enumerate}[(i)]
  \item $\heval{t}{\sigma}{v_1}{\sigma_1}$ and
    $\heval{t}{\sigma}{v_2}{\sigma_2}$ implies that $v_1 = v_2$ and
    that $\sigma_1 = \sigma_2$.
  \item $c \forward{\kappa \mapsto v} c_1$ and
    $c \forward{\kappa \mapsto v} c_2$ implies $c_1 = c_2$.
  \item $c \forward{O_1} c_1$ and $c \forward{O_2} c_2$ implies
    $O_1 = O_2$ and $c_1 = c_2$.
  \end{enumerate}
\end{lemma}
%

Causality now follows from Theorem~\ref{thr:productivity}  and
Lemma~\ref{lem:deterministic}.

\begin{corollary}[causality]
\label{cor:causality}
 Suppose (\ref{eq:causality:trans:seq}) as well as the following are 
 well-typed transition sequences 
 \[
   \stateI[\iota_0]{t} \forward{O_0'} \stateS[\iota_0']
  {N_0'}{\eta_0'}\forward{\kappa_0'\mapsto v_0'}\stateS[\iota_1']{N_0'}{\sigma_0'}
  \forward{O_1'}\stateS[\iota_1']{N_1'}{\eta_1'}\forward{\kappa_1'\mapsto v_1'}\dots
 \]
 Let $n \in \nats$ and suppose $\kappa_i' = \kappa_i$ and $v_i' = v_i$ for all $i<n$. 
 Then $O_n' = O_n$.
\end{corollary}

\subsubsection{Signal Independence}
\label{sec:signal-independence}

From the definition of the reactive semantics we can see that the
machine only updates an output channel $x : A \in \Gammaout$ if
it depends on the input value $\kappa \mapsto v$ that has just
arrived, i.e., if the machine is in a state
$\stateS{N}{\eta_N\stick{\kappa}{v}\eta_L}$ with
$\kappa \in \cl{N(x)}$. However, the typing system allows us to give
two useful \emph{static} criteria for when $\kappa \nin \cl{N(x)}$ is
guaranteed and thus the output signal $x$ need not (and indeed cannot)
be updated.

As alluded to earlier, values received on
buffered-only channels will never produce an output:
\begin{proposition}[buffered signal independence]
  \label{prop:bufSigIndependence}
  Suppose $t : \Delta \Rightarrow \Gammaout$ is a reactive program and
  \[
    \stateI[\iota_0]{t} \forward{O_0} \stateS[\iota_0]
    {N_0}{\eta_0}\forward{\kappa_0\mapsto v_0}\stateS[\iota_1]{N_0}{\sigma_0}
    \forward{O_1}\stateS[\iota_1]{N_1}{\eta_1}\forward{\kappa_1\mapsto v_1}\dots
  \]
  is a well-typed transition sequence for $t$. Then $O_{i+1}$ is empty
  whenever $\kappa_i :_{\cbuf} A \in \Delta$ for some $A$.
\end{proposition}

Secondly, the input context $\Delta$ for a given output signal
implementation gives us an upper bound on the push channels that will
trigger an update:
\begin{theorem}[push signal independence]
  \label{thr:sigIndependence}
  Suppose $\pair{t}{s} : \Delta \Rightarrow \Gammaout$ is a reactive
  program with $\Gammaout = \Gamma, z : C$ such that also
  $s : \Delta' \Rightarrow (z : C)$ is a reactive program for some
  $\Delta'\subset \Delta$ and
  \[
    \stateI[\iota_0]{\pair{t}{s}} \forward{O_0} \stateS[\iota_0]
    {N_0}{\eta_0}\forward{\kappa_0\mapsto v_0}\stateS[\iota_1]{N_0}{\sigma_0}
    \forward{O_1}\stateS[\iota_1]{N_1}{\eta_1}\forward{\kappa_1\mapsto v_1}\dots
  \]
  is a well-typed transition sequence for $\pair{t}{s}$. Then
  $z \mapsto v \in O_{i+1}$ implies that $\kappa_i \in
  \dom{\Delta'}$. In other words, the output channel $z$ is only
  updated when inputs in $\Delta'$ are updated.
\end{theorem}

\section{Metatheory}
\label{sec:metatheory}

In this section, we sketch the proof of the operational properties
presented in section~\ref{sec:main-results}, namely
Theorem~\ref{thr:productivity},
Proposition~\ref{prop:bufSigIndependence}, and
Theorem~\ref{thr:sigIndependence}. All three follow from
a more general semantic soundness property. To prove this property, we
first devise a semantic model of the Async RaTT calculus in the form
of a Kripke logical relation. That is, the model consists of a family
$\sem{A}(w)$ of sets of closed terms that satisfy the soundness
properties we are interested in. This family of sets is indexed by a
\emph{world} and is defined by induction on the structure of the type
$A$ and world $w$. The soundness proof is thus reduced to a proof that
$\hastype{}{t}{A}$ implies $t \in \sem{A}(w)$, which is also known as
the \emph{fundamental property} of the logical relation.

\subsection{Kripke Logical Relation}

The worlds $w$ for our logical relation consist of two components: a
natural number $n$ and a store $\sigma$. The number $n$ allows us to
model guarded recursive types via step-indexing
\citep{appel01indexed}. This is achieved by defining
$\sem{\DelayE A}(n + 1,\sigma)$ in terms of $\sem{A}(n,\sigma')$ for
some suitable $\sigma'$. Since recursive types $\Fix\,\alpha.A$ unfold
to $A[\DelayE(\Fix\,\alpha.A)/\alpha]$, we can define
$\sem{\Fix\,\alpha.A}(n+1,\sigma)$ in terms of $\sem{A}(n+1,\sigma)$
and $\sem{\Fix\,\alpha.A}(n,\sigma')$, which is well-founded since in
the former we refer to the smaller type $A$ and in the latter we refer
to a smaller step index $n$.

A key aspect of the operational semantics of Async RaTT is that it
stores delayed computations in a store $\sigma$. Hence, in order to
capture the semantics of a term $t$, we have to account for the fact
that $t$ may contain heap locations that point into some suitable
store $\sigma$. Intuitively speaking, the set $\sem{A}(n,\sigma)$
contains those terms that, starting with the store $\sigma$, can be
executed safely
to produce a value of type $A$. Ultimately, the index $\sigma$ enables us
to prove that the garbage collection performed by the reactive semantics
is indeed sound.

What makes $\sem{A}(n,\sigma)$ a Kripke logical relation is the fact
that we have a preorder $\lesssim$ on worlds such that
$(n,\sigma) \lesssim (n',\sigma')$ implies
$\sem{A}(n,\sigma) \subseteq \sem{A}(n',\sigma')$. We can think of
$(n',\sigma')$ as a future world reachable from $(n,\sigma)$, i.e., it
describes how the surrounding context changes as the machine performs
computations. There are four different kinds of changes, which we
address in turn below:

Firstly, time may pass, which means that we have fewer time steps left, i.e.,
$n > n'$.
Secondly, the machine performs garbage collection on the store
$\sigma$. The following garbage collection function describes this:
\[
  \gc{\eta_L}= \eta_L \hspace{2cm} \gc{\eta_N\stick{\kappa}{v}\eta_L}= \eta_L
\]
Third, the machine may store delayed computation in $\sigma$, which we
account for by the order $\heaple$ on heaps and stores:
\begin{mathpar}
  \inferrule*
  {\eta(l) = \eta'(l) \text{ for all } l \in \dom{\eta}}
  {\eta \heaple \eta'}
  \and
  \inferrule*
  {\eta_N \heaple \eta_N'\\ \eta_L \heaple \eta_L'}
  {\eta_N\stick{\kappa}{v}\eta_L\heaple \eta_N'\stick{\kappa}{v}\eta_L'}
\end{mathpar}
That is, $\sigma \heaple \sigma'$ iff $\sigma'$ is obtained from
$\sigma$ by storing additional terms.

Finally, the machine may receive an input value $\kappa \mapsto v$,
which is captured by the following order $\tickle$ on stores:
\begin{mathpar}
  \inferrule*
  {\sigma \heaple \sigma'}
  {\sigma \tickle \sigma'}
  \and
  \inferrule*
  {\eta_L \heaple \eta_L'\\ \kappa :_{\capa} A \in \Delta\\ \hastype[]{}{v}{A}}
  {\eta_L \tickle \eta_N\stick{\kappa}{v}\eta_L'}
\end{mathpar}
That is, in addition to the allocations captured by $\heaple$, the
order may also introduce an input value $\kappa \mapsto v$.

Taken together, we can define the Kripke preorder $\lesssim$ on worlds
as follows:
\[
  (n, \sigma) \lesssim (n', \sigma') \quad \text{iff} \quad n \ge
  n'\text{ and } \sigma \tickle \sigma'
\]
This does not include garbage collection, as it is restricted to
certain circumstances. Indeed, the machine performs garbage collection
only at certain points of the execution, namely at the end of an
output transition.

Finally, before we can give the definition of the Kripke logical
relation, we need to semantically capture the notion of input
independence that is needed both for the operational semantics of
$\select$ and the signal independence properties
(Proposition~\ref{prop:bufSigIndependence} and
Theorem~\ref{thr:sigIndependence}). In essence, we need that a heap
location $l$ in the world $(n+1,\sigma)$ should still be present in
the future world $(n,\sigma')$ in which we received an input on
a channel $\kappa \nin l$. We achieve this by making the logical
relation $\sem{\DelayE A}(n,\sigma)$ satisfy the following clock
independence property:
\[
\text{If }  l \in \sem{\DelayE A}(n,\sigma), \text{ then } l \in \sem{\DelayE A}(n,\restr[\cl{l}]{\sigma})
\]
where $\restr[\Theta]{\sigma}$ restricts $\sigma$ to heap locations
whose clocks are \emph{subclocks} of $\Theta$:
\begin{align*}
  \restr\eta(l) &= \eta(l) \qquad \text{ if } \cl{l} \subseteq
  \Theta\\
  \restr{\eta_N\stick{\kappa}{v} \eta_L} &=
  \begin{cases}
    \restr{\eta_N}\stick{\kappa}{v} \restr{\eta_L}
    &\text{if } \kappa\in \Theta\\
    \restr{\eta_L}
        &\text{if } \kappa\nin \Theta
      \end{cases}
\end{align*}

The full definition of the Kripke logical relation is given in
Figure~\ref{fig:log-rel}. In addition to the aspects discussed above,
it is parameterised by the context $\Delta$ and distinguishes between
the value relation $\vinterp{A}{w}$ and the term relation
$\tinterp{A}{w}$. The two relations are defined by well-founded
recursion by the lexicographic ordering on the tuple
$(n,\tsize{A},e)$, where $\tsize{A}$ is the size of $A$ defined below,
and $e = 1$ for the term relation and $e = 0$ for the value relation.
\begin{align*}
  \tsize{\alpha}
  &= \tsize{\DelayE A} = \tsize{\DelayA A} = \tsize{1} = \tsize{\Nat} =  1\\
  \tsize{A \times B}
  &= \tsize{A + B}= \tsize{A \to B} = 1 +
  \tsize{A} + \tsize{B}\\
  \tsize{\Box A}
  &= \tsize{\Fix\, \alpha. A} =  1 + \tsize{A}
\end{align*}
Note that in the definition for $\vinterp{\DelayE A}{w}$ in
Figure~\ref{fig:log-rel}, we use the shorthand $\sigma(l)$ for
$\eta_L(l)$, where $\eta_L$ is the later heap of $\sigma$.

\begin{figure}[t]
  \small
  \begin{align*}
  \vinterp{1}{w}
  &= \set{\unit},
  \\
  \vinterp{\Nat}{w}
  &= \setcom{\suc^n \, \zero}{ n \in \nats},
  \\
  \vinterp{A \times B}{w}
  &= \setcom{\pair{v_1}{v_2}}{v_1 \in \vinterp{A}{w}
    \land v_2 \in \vinterp{B}{w}},
  \\
  \vinterp{A + B}{w}
  &= \setcom{\interm_1 \, v}{v \in
    \vinterp{A}{w}} \cup
  \setcom{\interm_2 \, v}{v \in
    \vinterp{B}{w}}
  \\
  \vinterp{A \to B}{n,\sigma}
  &= \setcom{ \lambda x.t}{
  \forall \sigma' \tickge[\Delta] \gc{\sigma}, n'
    \le n, v \in \vinterp[\Delta]{A}{n',\sigma'}. 
    t[v/x] \in \tinterp[\Delta]{B}{n',\sigma'}}
  \\
  \vinterp{\Box A}{n,\sigma} &= \setcom{\rbox\,t}{t \in \tinterp{A}{n,\emptyset}}
  \\
  \vinterp{\DelayAS A}{0,\sigma} &=
  \setcom{\dfix\,x. t}{\dfix\,x. t \text{ a closed term}}\\
  \vinterp{\DelayAS A}{n+1,\sigma}
  &= \setcom{\dfix\,x.t}{ t[\dfix\,x.t/x] \in
    \tinterp{A}{n,\emptyset}}\\
  \vinterp{\DelayES A}{0,\sigma} &= \locs[\Delta] \cup \setcom{\inp[\kappa]}{\kappa :_\capa A \in \Delta, \capa \in \{\cpush, \cbp\}} \\
  \vinterp{\DelayES A}{n+1,\sigma}
  &= \setcom{l\in \locs[\Delta]}{
    \forall  \kappa \in \cl{l},
    \hastype[]{}{v}{\Delta(\kappa)} .  \sigma (l) \in
    \tinterp{A}{n,\restr[\cl{l}]{\clin{\gc{\sigma}}\!\!\stick{\kappa}{v}\!\!\clnin{\gc{\sigma}}}}
  }
  \\
  &\cup \setcom{\inp[\kappa]}{\kappa :_\capa A \in \Delta, \capa \in \{\cpush, \cbp\}}
  \\
  \vinterp{\Fix\,\alpha. A}{w}
  &= \setcom{\into\,v}{v \in \vinterp{A[\DelayES(\Fix\,\alpha.A)/\alpha]}{w}}
  \\[1em]
  \tinterp{A}{n,\sigma}
  &= \setcom{t}{t \text{ closed, } \forall \iota : \Delta.\forall \sigma' \tickge[\Delta] \sigma.\exists \sigma'', v
    . \heval{t}{\sigma'}{v}{\sigma''}
    \land v \in \vinterp[\Delta]{A}{n,\sigma''}}
\end{align*}
\vspace{-0.5em}
\begin{align*}
    \cinterp{\cdot}{w}  &= \set{\ast}
  \\
  \cinterp{\Gamma,x:A}{w}
  &= \setcom{\gamma[x \mapsto v]}
  {\gamma \in \cinterp{\Gamma}{w},v
    \in\vinterp{A}{w}}
  \\
  \cinterp{\Gamma,\tick[\theta]}{n,\eta_N\stick{\kappa}{v}\eta_L}
  &= \setcom{\gamma \in \cinterp{\Gamma}{n+1,(\eta'_N,\clnin{\eta'_L})}}{
    \begin{aligned}
      & \restr{\eta_L}=\restr{\eta'_L},\restr{\eta_N}=\restr{\eta'_N},\\
      &\eta'_N \in
      \heaps[\kappa],\hastype[]{}{v}{\Delta(\kappa)},\kappa \in
      \Theta,\text{ and }\\&\Theta \subseteq \domp{\Delta}, \text{ where } \Theta = \cleval{\theta\gamma}
    \end{aligned}
  }
\end{align*}
\center{\textbf{Using}}
\begin{align*}
 \domp{\Delta} & = \setcom{\kappa}{ \exists A, \capa \in\{\cpush,\cbp\}. x :_{\capa} A \in \Delta} & 
 \locs[\Delta] & = \setcom{l \in \locs}{\cl l \subseteq \domp{\Delta}} 
\end{align*}
\caption{Logical relation.}
  \label{fig:log-rel}
\end{figure}

Our goal is to prove the fundamental property, i.e., that
$\hastype{}{t}{A}$ implies $t \in \tinterp{A}{n,\sigma}$, by induction
on the typing derivation. Therefore, we need to generalise the
fundamental property to open terms as well. That means we also need a
corresponding logical relation for contexts, which is given at
the bottom of Figure~\ref{fig:log-rel}. The interpretation of
$\tick[\theta]$ in a context is quite technical, but is essentially
determined by the interpretation of $\DelayE$ due to the requirement
of being left adjoint~\cite{birkedal2020modal}.

The three logical relations indeed satisfy the preservation under the
Kripke preorder $\lesssim$.
\begin{lemma}
  \label{lem:vinterp_weakening}
  Let $n \ge n'$ and $\sigma \tickle \sigma'$.
  \begin{enumerate}[(i)]
  \item
    $\vinterp{A}{n,\sigma} \subseteq
    \vinterp[\Delta]{A}{n',\sigma'}$.
  \item
    $\tinterp{A}{n,\sigma} \subseteq
    \tinterp[\Delta]{A}{n',\sigma'}$.
      \item
    $\cinterp{\Gamma}{n,\sigma} \subseteq
    \cinterp[\Delta]{\Gamma}{n',\sigma'}$.
  \end{enumerate}
\end{lemma}
Preservation under garbage collection, however, only holds for values
and tick-free contexts:
\begin{lemma}(garbage collection)
  \label{lem:vinterp_gc}
  \begin{enumerate}[(i)]
  \item
    $\vinterp{A}{n,\sigma} \subseteq
    \vinterp{A}{n,\gc{\sigma}}$.
  \item
    $\cinterp{\Gamma}{n,\sigma} \subseteq
    \cinterp{\Gamma}{n,\gc{\sigma}}$ if $\Gamma$ is tick-free.
  \end{enumerate}
\end{lemma}
Moreover, the clock independence property holds for both the value
and context relations:
\begin{lemma}~%
  \label{lem:vinterp_restr}%
  \begin{enumerate}[(i)]
  \item If $v \in \vinterp{\DelayE A}{n,\sigma}$, then
    $v \in \vinterp{\DelayE A}{n,\sigma'}$, for any $\sigma'$ 
    with $\restr[\cl v]{\sigma} = \restr[\cl v]{\sigma'}$.
  \item If $\gamma \in \cinterp{\Gamma,\tick[\theta]}{n,\sigma}$,
    then $\gamma \in \cinterp{\Gamma,\tick[\theta]}{n,\restr[\cleval{\theta\gamma}]\sigma}$
  \end{enumerate}
\end{lemma}

Finally, we obtain the soundness of the language by the following
fundamental property of the logical relation:
\begin{theorem}
  \label{thr:lrl}
  Given $\wfcxt{\Gamma}$, $\hastype{\Gamma}{t}{A}$, and
  $\gamma \in \cinterp{\Gamma}{n,\sigma}$, then
  $t\gamma \in \tinterp{A}{n,\sigma}$.
\end{theorem}
The proof is a standard induction on the typing relation
$\hastype{\Gamma}{t}{A}$ that makes use of the aforementioned closure
properties of the logical relations and is included in
Appendix~\ref{sec:proof-fund-prop}.

\subsection{Operational Properties}

We close this section by showing how we can use the fundamental
property to prove the operational properties presented in
section~\ref{sec:main-results}. To this end, we will sketch the proofs
of Theorem~\ref{thr:productivity},
Proposition~\ref{prop:bufSigIndependence}, and
Theorem~\ref{thr:sigIndependence}.

\subsubsection{Productivity}

In the following we assume a fixed reactive interface
$\Delta \Rightarrow \Gammaout$, for which we define the following sets
of machine states $\pIState{n}$ and $\pState{n}$ for the reactive
semantics:
\begin{align*}
  \pIState{n} &=\setcom{\stateI{t}}{\iota : \Delta \land t \in
    \tinterp{\Prod{\Gammaout}}{n,\emptyset}}\\
  \pState{n} &=\setcom{\stateS{N}{\eta}}{\iota : \Delta \land
    \forall x \mapsto l \in N.\exists x :
    A \in \Gammaout.l \in
    \vinterp{\DelayE (\Sig\,A)}{n,\eta}}
\end{align*}

The following lemma proves that the machine stays inside the sets of
states defined above and will only produce well-typed outputs. For the
latter, we make use of the fact that $v \in \vinterp{A}{n,\sigma}$ iff
$\hastype[]{}{v}{A}$ for every value type $A$.

\begin{lemma}[productivity]~
  \label{lem:productivity}
  \begin{enumerate}[(i)]
  \item If $t : \Delta \Rightarrow \Gammaout$ and $\iota : \Delta$, then
    $\stateI{t} \in \pIState{n}$ for all $n \in \nats$.%
    \label{item:productivityA}
  \item If $\stateI{t} \in \pIState{n}$, then there is a transition
    $\stateI{t} \forward{O} \stateS{N}{\eta}$ such that
    $\stateS{N}{\eta} \in \pState{n}$ and $\vdash O : \Gammaout$.%
    \label{item:productivityB}
  \item If $\stateS{N}{\eta} \in \pState{n+1}$ and
    $\vdash \kappa \mapsto v : \Delta$, then there is a sequence of
    two transitions
    \[\stateS{N}{\eta} \forward{\kappa \mapsto
      v}\stateS[\iota']{N}{\sigma}\forward{O} \stateS[\iota']{N'}{\eta'}\] such
    that $\stateS[\iota']{N'}{\eta'} \in \pState{n}$ and
    $\vdash O : \Gammaout$.
    \label{item:productivityC}
  \end{enumerate}
\end{lemma}
\begin{proof}
  ~
  \begin{enumerate}[(i)]
  \item We need to show that
    $t \in \tinterp{\Prod{\Gammaout}}{n,\emptyset}$.  Since
    $t : \Delta \Rightarrow \Gammaout$, we know that
    $\hastype[\Delta]{}{t}{\Prod{\Gammaout}}$. Hence, by
    Theorem~\ref{thr:lrl}, we have that
    $t \in \tinterp{\Prod{\Gammaout}}{n,\emptyset}$.
  \item Let $\stateI{t} \in \pIState{n}$. We thus have
    $t \in \tinterp{\Prod\Gammaout}{n,\emptyset}$. Therefore, by definition, we have
    $\heval{t}{\emptyset}{\tuple{v_1\cons w_1, \dots, v_k \cons
        w_k}}{\eta}$ with $v_i \in \vinterp{A_i}{n,\eta}$ and
    $w_i \in \vinterp{\DelayE (\Sig\,A_i)}{n,\eta}$. Since $\Sig\,A_i$
    is not a value type, we know that each $w_i$ must be a heap
    location, so we can write $l_i$ for $w_i$. Hence, by definition,
    $\stateI{t} \forward{x_1 \mapsto v_1,\dots,x_k\mapsto v_k}
    \stateS{x_1 \mapsto l_1,\dots,x_k\mapsto l_k}{\eta}$ and
    $\stateS{x_1 \mapsto l_1,\dots,x_k\mapsto l_k}{\eta} \in
    \pState{n}$. 
    Since each $A_i$ is a value type, we know that
    $v_i \in \vinterp{A_i}{n,\eta}$ implies $\hastype[]{}{v_i}{A_i}$
    and thus
    $\vdash x_1 \mapsto v_1,\dots,x_k\mapsto v_k : \Gammaout$.

\item By definition
  $\stateS{N}{\eta} \forward{\kappa \mapsto v}\stateS[\iota']{N}{\sigma}$, where
  $\sigma = \clin{\eta} \stick{\kappa}{v} \clnin\eta$. We claim the following three statements:
  \begin{align}
    &\text{If } y : B \in \Gammaout, N(y) = l, \kappa \in \cl{l}, \text{ then } \adv\,l \in
    \tinterp{\Sig\,B}{n,\sigma}\label{eq:productivityC1}\\
    &\text{If } y : B \in \Gammaout, N(y) = l, \kappa \nin \cl{l}, \text{ then } l \in
    \vinterp{\DelayE (\Sig\,B)}{n,\sigma}
    \label{eq:productivityC2}\\
    &\text{If } \sigma' \heapge \sigma, N_1 \subseteq N, \text{ then }
    \stateS[\iota']{N_1}{\sigma'} \forward{O} \stateS[\iota']{N_2}{\eta'},
    \stateS[\iota']{N_2}{\eta'} \in \pState{n}, \gc{\sigma'} \heaple \eta', \vdash O :
    \Gammaout
    \label{eq:productivityC3}
  \end{align}
  From \eqref{eq:productivityC3}, we then obtain that
  $\stateS[\iota']{N}{\sigma} \forward{O} \stateS[\iota']{N'}{\eta'}$, with
  $\stateS[\iota']{N'}{\eta'} \in \pState{n}$ and $\vdash O : \Gammaout$. We conclude
  this proof by proving the above three claims:
  \begin{itemize}
  \item Proof of \eqref{eq:productivityC1}. Since
    $\stateS{N}{\eta} \in \pState{n+1}$, we know that
    $l \in \vinterp{\DelayE(\Sig\,B)}{n+1,\eta}$ and thus
    $\adv\,l \in \tinterp{\Sig\,B}{n,\sigma}$.
  \item Proof of \eqref{eq:productivityC2}. Since
    $\stateS{N}{\eta} \in \pState{n+1}$, we know that
    $l \in \vinterp{\DelayE(\Sig\,B)}{n+1,\eta}$. 
    Since $\kappa \nin \cl{l}$, we can conclude that
    $\restr[\cl l]{\eta} = \restr[\cl l]{\clnin{\eta}}$. In turn, by
    Lemma~\ref{lem:vinterp_restr}, this means that
    $l \in \vinterp{\DelayE(\Sig\,B)}{n +1,\clnin{\eta}}$,
    which by Lemma~\ref{lem:vinterp_weakening} implies that
    $l \in \vinterp{\DelayE(\Sig\,B)}{n,\sigma}$.
  \item Proof of \eqref{eq:productivityC3}. We proceed by induction on
    $N_1$. The case $N_1 = \cdot$ is trivial.  For the case
    $N_1 = y \mapsto w, N'_1$, we distinguish between the case where
    $\kappa \in \cl w$ and the case where $\kappa \nin \cl w$. Then we
    can apply \eqref{eq:productivityC1} or \eqref{eq:productivityC2},
    respectively, and use the induction hypothesis for $N'_1$ and some
    $\sigma'' \heapge \sigma'$. If $\kappa \nin \cl w$, then
    $\sigma''$ is just $\sigma'$. If $\kappa \in \cl w$, then
    $\sigma''$ arises from the evaluation
    $\heval[\iota']{\adv\,w}{\sigma'}{v\cons w'}{\sigma''}$ we obtained from
    \eqref{eq:productivityC1}. In either case, we make use of
    Lemma~\ref{lem:vinterp_weakening} and Lemma~\ref{lem:vinterp_gc}
    to conclude $\stateS[\iota']{N_2}{\eta'} \in \pState{n}$. \popQED
  \end{itemize}
\end{enumerate}
\end{proof}

The productivity property is now a straightforward consequence of the
above lemma:
\begin{proof}[proof of Theorem~\ref{thr:productivity} (productivity)]
  For each $n \in \nats$ we can we can construct the following finite transition
  sequence $s_n$ using Lemma~\ref{lem:productivity}:
  \[
    \stateI[\iota_0]{t} \forward{O_0}
    \stateS[\iota_0]{N_0}{\eta_0}\forward{\kappa_0\mapsto v_0}\stateS[\iota_1]
    {N_0}{\sigma_0} \forward{O_1}\stateS[\iota_1]
    {N_1}{\eta_1}\forward{\kappa_1\mapsto v_1}\dots
    \forward{O_n}\stateS[\iota_n]{N_n}{\eta_n}
  \]
  with $\vdash O_i : \Gammaout$ for all $0\le i \le n$. By
  Lemma~\ref{lem:deterministic}, $s_n$ is a prefix of $s_m$ for all
  $m > n$. We thus obtain the desired infinite transition sequence as
  the limit of all $s_n$.
\end{proof}

\subsubsection{Signal Independence}
Proposition~\ref{prop:bufSigIndependence} is a straightforward
consequence of Lemma~\ref{lem:productivity}:
\begin{proof}[Proof of Proposition~\ref{prop:bufSigIndependence}
  (buffered signal independence)]
  By Lemma~\ref{lem:productivity}, for any pair of transitions
  $\stateS[\iota']{N_i}{\eta_i} \forward{\kappa_i \mapsto
    v_i}\stateS[\iota_{i+1}]{N_i}{\sigma_i}\forward{O_{i+1}}
  \stateS[\iota_{i+1}]{N_{i+1}}{\eta_{i+1}}$ in a sequence starting
  from $\stateI[\iota_i]{t}$, we have that
  $\stateS[\iota']{N_i}{\eta_i} \in \pState{1}$. In particular, this
  means that $l \in \vinterp{\DelayE (\Sig\,A)}{1,\eta_i}$ for any
  $x \mapsto l \in N_i$ with $x : A \in \Gammaout$, which in turn
  means that $l \in \locs[\Delta]$. Hence, $\kappa_i \notin \cl l$,
  and so $O_{i+1}$ is empty.
\end{proof}

For the proof of Theorem~\ref{thr:sigIndependence}, we first define
the following sets of machine states in the context of a partial map
$\Delta'$ that maps variables $x$ to input contexts $\Delta'_x$:
\begin{align*}
  \pTState n &=\setcom{\stateS{N}{\eta}}{\iota : \Delta \land
    \forall x \mapsto w \in N. \Delta'_x \subseteq \Delta \land \exists x :
    A \in \Gammaout.
    w \in \vinterp[\Delta'_x]{\DelayE (\Sig\,A)}{n,\eta}}
\end{align*}

Machine states in $\pTState n$ are maintained during the execution of
the machine:
\begin{lemma} \label{lem:stateT}
 If $\stateS{N}{\eta}\in \pTState{n+1}$ and  $\stateS{N}{\eta} \forward{\kappa \mapsto
      v}\stateS[\iota']{N}{\sigma}\forward{O}
    \stateS[\iota']{N'}{\eta'}$, then
    $\stateS[\iota']{N'}{\eta'} \in \pTState{n}$.
\end{lemma}
\begin{proof}
  Note that $\sigma = \clin{\eta}\stick{\kappa}{v}\clnin{\eta}$ and
  that $\clnin{\eta} \subseteq \eta'$.  Suppose $(x \mapsto w) \in N$,
  and note that $w$ must be a location, since $\Sig \, A$ is not a
  value type. We will write $l$ for $w$ to emphasise this.  Then there
  is an $l'$ such that $(x \mapsto l')\in N'$ and we must show that
  $l'\in\vinterp[\Delta'_x]{\DelayE (\Sig\,A)}{n,\eta'}$ using the
  hypothesis $l \in \vinterp[\Delta'_x]{\DelayE (\Sig\,A)}{n+1,\eta}$.
  Suppose first that $\kappa \notin \cl l$. Then $l'=l$ and thus
  \begin{align*}
    l'\in \vinterp[\Delta'_x]{\DelayE (\Sig\,A)}{n,\restr[\cl l]{\eta}} 
    \subseteq \vinterp[\Delta'_x]{\DelayE (\Sig\,A)}{n,\clnin{\eta}} 
    \subseteq \vinterp[\Delta'_x]{\DelayE (\Sig\,A)}{n,\eta'} 
  \end{align*}

 Suppose now $\kappa \in \cl l$. In that case, $l'$ must have occurred by evaluating
 $\heval{\sigma'(l)}{\sigma'}{w \cons l'}{\sigma''}$ for some $\sigma', \sigma''$ such that $\sigma\subseteq \sigma'$, and 
 $\gc{\sigma''}\subseteq\eta'$. The hypothesis tells us that 
\begin{align*}
 \eta(l) \in \tinterp[\Delta'_x]{\DelayE (\Sig\,A)}{n,\clin\eta\stick{\kappa}{v} \clnin{\eta}}
\end{align*}
 Since $\clin\eta\stick{\kappa}{v} \clnin{\eta} \subseteq \sigma'$, this means that $\heval[\iota]{\eta(l)}{\sigma'}{w' \cons l''}{\sigma'''}$.  
 Since $\eta(l) = \clin\eta(l) = \sigma'(l)$, by the determinism of
 the operational semantics (Lemma~\ref{lem:deterministic}),
 $\sigma''' = \sigma''$ and $l'' = l'$. From this we conclude that
\[
  l'  \in \vinterp[\Delta'_x]{\DelayE (\Sig\,A)}{n,\sigma''} 
   = \vinterp[\Delta'_x]{\DelayE (\Sig\,A)}{n,\gc{\sigma''}} 
   \subseteq \vinterp[\Delta'_x]{\DelayE (\Sig\,A)}{n,\eta'}\qedhere
   \popQED
\]
\end{proof}

\begin{proof}[Proof of Theorem~\ref{thr:sigIndependence} (push signal independence)]
  Let $\Gammaout = x_1 : B_1, \dots, x_m : B_m, z : C$. The
  initialisation transition
 \[
 \stateI[\iota_0]{\pair ts}\forward{x_1\mapsto
        v_1,\dots,x_m\mapsto v_m, z \mapsto w} \stateS[\iota_0]{x_1\mapsto
        l_1,\dots,x_m\mapsto l_m, z \mapsto l}{\eta_0} = \stateS[\iota_0]{N_0}{\eta_0}
 \]
 is caused by evaluations of the form
 $\heval t\emptyset {\tuple{ v_1\cons l_1,\dots,v_m \cons
     l_m}}{\eta}$, and $\heval s {\eta}{w \cons l}{\eta_0}$.  By
 assumption $\hastype[\Delta']{}{s}{\Sig \,C}$ and thus
 $s \in \tinterp[\Delta']{\Sig\,C}{n,\emptyset}$ for all $n$ by
 Theorem~\ref{thr:lrl}, which in turn implies
 $l \in \vinterp[\Delta']{\DelayE\Sig\,C}{n,\eta_0}$ for all $n$.
 Likewise $l_j \in \vinterp[\Delta]{\DelayE\Sig\,B_j}{n,\eta_0}$ for
 all $n$, and $j = 1, \dots ,m$.  So
 $\stateS[\iota_0]{x_1\mapsto l_1,\dots,x_m\mapsto l_m, z \mapsto
   l}{\eta_0} \in \pTState[\Delta'']{n}$ for all $n$, where
 $\Delta''_z=\Delta'$ and $\Delta''_{x_j}= \Delta$, for
 $j = 1, \dots, m$. By $n$ applications of Lemma~\ref{lem:stateT}, we
 thus obtain that
 $\stateS[\iota_n]{N_n}{\eta_n}\in \pTState[\Delta'']{1}$. In
 particular, $N_n(z) = l'$ for some
 $l' \in \vinterp[\Delta']{\DelayE (\Sig\,C)}{1,\eta_n}$. Hence,
 $\cl{l'} \subseteq \dom{\Delta'}$ and thus $z \mapsto v \in O_{i+1}$
 implies $\kappa_i\in\cl{l'}$, which in turn implies
 $\kappa_i \in \dom{\Delta'}$.
\end{proof}

\section{Related work}
\label{sec:related-work}

Functional reactive programming originates with \citet{FRAN}.
The use of modal types for FRP was first suggested by 
\citet{krishnaswami2011ultrametric}, and the connection between linear temporal logic and 
FRP was discovered independently by~\citet{jeffrey2012} and~\citet{Jeltsch2012}. Although 
some of these calculi have been implemented, they do not offer operational guarantees like the
ones proved here for lack of space leaks. The first such operational guarantees were given by
\citet{krishnaswami2012higher} who describe a modal FRP language using linear types and 
allocation resources to statically bound the memory used by a reactive program. The simpler, but 
less precise, idea of using an aggressive garbage collection technique for avoiding space
leaks is due to \citet{krishnaswami13frp}. Krishnaswami's calculus used a dual context
approach to programming with modal types. \citet{bahr2019simply} recast these results 
in a Fitch-style modal calculus, the first in the RaTT family. This was later implemented in
Haskell with some minor modifications~\cite{rattusJFP}. 

All the above calculi are based on a global notion of time, which in
almost all cases is discrete.  In particular, the modal operator
$\Delay$ for time steps in these calculi refers to the next time step
on the global clock.  One can of course also understand the step
semantics of Async RaTT as operating on a global clock, but in our
model each step is associated with an input coming from an input
channel, and this allows us to define the delay modality $\DelayE$ as
a delay on a set of input channels.  From the model perspective,
$\DelayE A$ carries some similarities with the type
$\Delay(\Diamond A)$, where $\Diamond A \iso A + \Delay\Diamond A$ is
a guarded recursive type.  This encoding, however, suffers from
the efficiency and abstraction problems mentioned in the
introduction.

The only asynchronous modal FRP calculus that we are aware of is
$\lwid$ defined by \citet{graulund2021adjoint}, which takes $\Diamond$
as a type constructor primitive and endows it with synchronisation
primitive similar to $\select$ in Async RaTT. However, the programming
primitives in $\lwid$ are very different from the ones use here.  For
example, $\lwid$ allows an element of $\Diamond A$ to be decomposed
into a time and an element of $A$ at that time, and much programming
with $\Diamond$ uses this decomposition. There is also no delay type
constructor $\Delay$, so $\DelayE$ is not expressible: Unlike
$\DelayE A$, an element of $\Diamond A$ could give a value of type $A$
already now. Graulund et al. provide a denotational semantics for
$\lwid$, but no operational semantics, and no operational guarantees
as proved here.

Another approach to avoiding space leaks and non-causal reactive
programs is to devise a carefully designed interface to manipulate
signals such as Yampa~\citep{nilsson02yampa} or
FRPNow!~\citep{ploeg15FRPNow}. Rhine~\cite{barenz18Rhine} is a recent
refinement of Yampa that annotates signal functions with type-level
clocks, which allows the construction of complex dataflow graphs that
combine subsystems running at different clock speeds. The typing
discipline fixes the clock of each subsystem \emph{statically} at
compile time, since the aim of Rhine is provide efficient resampling
between subsystems. By contrast, the type-level clocks of Async RaTT
are existentially quantified, which allows Async RaTT programs to
\emph{dynamically} change the clock of a signal, e.g., by using the
\varid{switch} combinator from \autoref{sec:synchr-comb}.

\citet{push-pull} proposed a \emph{push-pull} implementation of FRP,
where signals (which in the tradition of classic FRP~\citep{FRAN} are
called behaviours) are updated at discrete time steps (push), but can
also be sampled at any time between such updates (pull). We can represent
such push-pull signals in Async RaTT using the type
$\Sig\,(\sym{Time} \to A)$, i.e., at each tick of the clock we get a
new function $\sym{Time} \to A$ that describes the time-varying value
of the signal until the next tick of the clock.

Futures, first implemented in MuliLisp~\citep{halstead85Multilisp} and
now commonly found in many programming languages under different names
(promise, async/await, delay, etc.), provide a powerful abstraction to
facilitate communication between concurrent computations. A value of
type $\sym{Future}\, A$ is the promise to deliver a value of type $A$
at some time in the future. For example, a function to read the
contents of a file could immediately return a value of type
$\sym{Future}\, \varid{Buffer}$ instead of blocking the caller until
the file was read into a buffer. Async RaTT can provide a similar
interface using the type modality $\DelayE$, either directly or by
defining $\sym{Future}$ as a guarded recursive type
$\sym{Future}\, A \iso  A + \DelayE(\sym{Future}\,A)$ to give
$\sym{Future}$ a monadic interface. Since Async RaTT does not require
the set of push-only channels to be finite, we could implement a
function that takes a filename $f$ and returns a result of type
$\sym{Future}\, \sym{Buffer}$ simply as a family of channels
$\sym{readFile}_f :_\cpush \varid{Buffer}$. The machine would monitor
delayed computations for clocks containing these channels, initiate
reading the corresponding files in parallel, and provide the value of
type \varid{Buffer} on the channel upon completion of the file reading
procedure.

As mentioned earlier, \citet{krishnaswami2012higher} used a linear
typing discipline to obtain static memory bounds. In addition to such
memory bounds, synchronous (dataflow) languages such as
Esterel~\citep{berry1984esterel}, Lustre~\citep{caspi1987lustre}, and
Lucid Synchrone~\citep{pouzet2006lucid} even provide bounds on
runtime. Despite these strong guarantees, Lucid Synchrone affords a
high-level, modular programming style with support for higher-order
functions. However, to achieve such static guarantees, synchronous
dataflow languages must necessarily enforce strict limits on the
dynamic behaviour, disallowing both time-varying values of arbitrary
types (e.g., we cannot have a stream of streams) and dynamic switching
(i.e., no functionality equivalent to the \varid{switch}
combinator). Both Lustre and Lucid Synchrone have a notion of a clock,
which is simply a stream of Booleans that indicates at each tick of
the global clock, whether the local clock ticks as well.


\section{Conclusion and future work}
\label{sec:concl-future-work}

This paper presented Async RaTT, the first modal language for 
asynchronous FRP with operational guarantees. We showed how the 
new modal type $\DelayE$ for asynchronous delay can be used to 
annotate the runtime system with dependencies from output channels
to input channels, ensuring that outputs are only recomputed when necessary.
The examples of the integral and the derivative even show how the 
programmer can actively influence the update rate of output channels.

The choice of Fitch-style modalities is a question of taste, and
we believe that the results could be reproduced in a dual context 
language. Even though Fitch-style uses non-standard operations
on contexts, other languages in the RaTT family have been implemented
as libraries in Haskell~\cite{rattusJFP}. We therefore believe that 
also Async RaTT can be implemented in Haskell or other 
functional programming languages, giving programmers
access to a combination of features from RaTT and 
the hosting programming language. 

One aspect missing from Async RaTT is filtering of output channels.
For example, it is not possible to write a filter function that only produces
output when some condition on the input is met. The best way to
do model this is using an output channel of type $\sym{Maybe}(A)$,
leaving it to the runtime system to only push values of type $A$ to the
consumers of the output channel. This way the filtering is external
to the programming language. We see no way to meaningfully extend 
the runtime model of Async RaTT to internalise it. 

\begin{acks}
  M{\o}gelberg was supported by the Independent Research Fund Denmark
  grant number 2032-00134B. 
\end{acks}

\bibliographystyle{ACM-Reference-Format}
\bibliography{paper}

\input{appendix}


\end{document}

%% file: prelude.tex
\usepackage{amsmath,amsfonts,amsthm,stmaryrd}
\usepackage{xspace}
\usepackage{thmtools}
\usepackage{mathpartir}
\usepackage{natbib}
\usepackage[shortlabels,inline]{enumitem}
\theoremstyle{definition}\newtheorem{definition}{Definition}[section]
\newtheorem{theorem}[definition]{Theorem}
\theoremstyle{theorem}
\theoremstyle{theorem}\newtheorem{lemma}[definition]{Lemma}
\theoremstyle{theorem}\newtheorem{proposition}[definition]{Proposition}
\theoremstyle{theorem}\newtheorem{corollary}[definition]{Corollary}
\theoremstyle{example}
\theoremstyle{remark}
\usepackage{hyperref}
\usepackage{tikz-cd}
\usepackage{tikz}
\usetikzlibrary{arrows}
\usepackage{bussproofs}
\EnableBpAbbreviations 
\usepackage{color}
\usepackage{algorithmicx}
\usepackage{algpseudocode}
\usepackage{algorithm}
\usepackage{epigraph}
\usepackage{suffix}
\usepackage{mathtools}
\usepackage{subcaption}

\newcommand{\lwid}{\ensuremath{\lambda_{\sym{widget}}}\xspace}

\newcommand\locs[1][]{\ensuremath{\sym{Loc}_{#1}}\xspace}
\newcommand\heaps[1][]{\ensuremath{\sym{Heap}^{#1}}\xspace}

\newcommand\chans{\ensuremath{\sym{Chan}}\xspace}

\newcommand\inp[1][\kappa]{\sym{wait}_{#1}}
\newcommand\buf[1][\kappa]{\sym{read}_{#1}}

\newcommand\mktick[2][\kappa\mapsto v]{\sym{tick}_{#1}\left(#2\right)}
\newcommand\cl[1]{\sym{cl}\left(#1\right)}
\newcommand{\stick}[2]{\left\langle #1 \mapsto #2 \right\rangle}

\newcommand{\synch}{\sym{sync}}

\newcommand\nin{\not\in}
\newcommand\sym[1]{\mathsf{#1}}
\newcommand\tuple[1]{\left\langle#1\right\rangle}
\newcommand\pair[2]{\left(#1,#2\right)}
\newcommand\set[1]{\left\{#1\right\}}
\newcommand{\setcom}[2]{\set{#1\left\vert\vphantom{#1}\,#2\right.}}

\newcommand{\tick}[1][]{
  \ifthenelse{\equal{#1}{}}
  {\checkmark}
  \checkmark_{\hspace{-4pt}#1\hspace{1pt}}
}

\newcommand\dom[1]{\sym{dom}\left(#1\right)}
\newcommand\domp[1]{\sym{dom}_{\cpush}\left(#1\right)}

\newcommand\tickle[1][\Delta]{\mathop{\sqsubseteq_\checkmark^{#1}}}
\newcommand\tickge[1][\Delta]{\mathop{\sqsupseteq_\checkmark^{#1}}}
\newcommand\heaple{\sqsubseteq}
\newcommand\heapge{\sqsupseteq}

\newcommand\allocate[2][\Theta]{\sym{alloc}^{#1}\left(#2\right)}

\newcommand\varid[1]{\ensuremath{\mathit{#1}}}

\newcommand{\interm}{\sym{in}}
\newcommand{\caseterm}[5]{\sym{case}\, #1 \, \sym{of}\, \interm_1\,#2
  . #3 ; \interm_2\,#4 . #5}
\newcommand{\caset}[1]{\sym{case}\, #1 \, \sym{of}}
\newcommand{\letterm}[3]{\sym{let}\, #1 = #2 \, \sym{in} \, #3}

\newcommand\case{\mathsf{case}}
\newcommand\unit{\langle\rangle}
\newcommand\zero{0}
\newcommand\suc{\sym{suc}}

\newcommand\Fix{\sym{Fix}}
\newcommand\fix{\sym{fix}}
\newcommand\dfix{\sym{dfix}}

\newcommand\adv{\sym{adv}}
\newcommand\select{\sym{select}}
\newcommand\delay{\sym{delay}}
\newcommand\rbox{\sym{box}}

\newcommand\unbox[1][]{\sym{unbox}^{#1}}
\newcommand\never{\sym{never}}
\newcommand\out{\sym{out}}
\newcommand\into{\sym{into}}

\newcommand{\cons}{::}

\newcommand{\gc}[1]{\sym{gc}(#1)}

\newcommand\recN[5]{\sym{rec}_{\Nat}(#1,#2\,#3.#4,#5)}

\newcommand{\stabilize}[1]{#1^\Box}

\newcommand{\cpush}{\sym{p}}
\newcommand{\cbuf}{\sym{b}}
\newcommand{\cbp}{\sym{bp}}
\newcommand{\capa}{c}

\makeatletter
\newcommand{\superimpose}[2]{{%
  \ooalign{%
    \hfil$\m@th#1\@firstoftwo#2$\hfil\cr
    \hfil$\m@th#1\@secondoftwo#2$\hfil\cr
  }%
}}
\makeatother

\newcommand\Delay{\ensuremath{{\bigcirc}}}
\newcommand\DelayE{\ensuremath{\mathrlap{\hspace{2.1pt}{\scriptstyle\exists}}{\Delay}}}
\newcommand\DelayA{\ensuremath{\mathrlap{\hspace{2.1pt}{\scriptstyle\forall}}{\Delay}}}
\newcommand\DelayES{\ensuremath{\mathrlap{\hspace{1.7pt}{\scriptstyle\exists}}{\Delay}}}
\newcommand\DelayAS{\ensuremath{\mathrlap{\hspace{1.7pt}{\scriptstyle\forall}}{\Delay}}}

\newcommand\stable{\sym{stable}}

\newcommand\Sig{\sym{Sig}}

\newcommand\Nat{\sym{Nat}}

\newcommand\Float{\sym{Float}}
\newcommand\Unit{\sym{1}}

\newcommand\nats{{\mathbb N}}

\newcommand\Prod[1]{\sym{Prod}\left(#1\right)}

\newcommand\Gammaout{\Gamma_{\sym{out}}}

\newcommand\vinterp[3][\Delta]{\mathcal{V}_{#1}\llbracket #2\rrbracket(#3)}
\newcommand\tinterp[3][\Delta]{\mathcal{T}_{#1}\llbracket #2\rrbracket(#3)}
\newcommand\cinterp[3][\Delta]{\mathcal{C}_{#1}\llbracket #2\rrbracket(#3)}
\newcommand\sem[1]{\llbracket #1 \rrbracket}

\newcommand\wfcxt[2][\Delta]{#2 \vdash_{#1} }

\newcommand\hastype[4][\Delta]{#2 \vdash_{#1} #3 : #4}
\newcommand\isclock[3][\Delta]{#2 \vdash_{#1} #3 : \sym{Clock}}

\newcommand\stateS[3][\iota]{\left\langle #2;#3 ; #1 \right\rangle}
\newcommand\stateI[2][\iota]{\left\langle #2 ; #1 \right\rangle}
\renewcommand\state[2]{\left\langle #1;#2 \right\rangle}
\newcommand\heval[5][\iota]{ \state{#2}{#3} \Downarrow^{#1} \state{#4}{#5}}
\newcommand\cleval[1]{\left|#1\right|}

\newcommand\restr[2][\Theta]{\left[#2\right]_{#1}}
\newcommand\clin[2][\kappa]{\left[#2\right]_{#1 \in}}
\newcommand\clnin[2][\kappa]{\left[#2\right]_{#1 \nin}}

\newcommand{\pIState}[1]{I_{#1}}
\newcommand{\pState}[1]{S_{#1}}
\newcommand{\pTState}[2][\Delta']{T_{#2}^{#1}}

\newcommand\tsize[1]{\left| #1\right|}

\newcommand\forward[1]{\stackrel{#1}{\Longrightarrow}}

\newcommand{\iso}{\cong}


%% file: appendix.tex
\newpage
\appendix
\newenvironment{thmcpy}[1]{\noindent\textbf{#1.}\itshape}{}
\newenvironment{thmcpy*}[2]{\noindent\textbf{#1}~(#2)\textbf{.}\itshape}{}
\section{Proof of Fundamental Property}
\label{sec:proof-fund-prop}

Given a heap $\eta$ we use the following notation to construct a
well-formed store with $\stick{\kappa}{v}$ as follows:
\[
  \mktick{\eta}= \clin{\eta} \stick{\kappa}{v} \clnin{\eta}
\]

\begin{lemma}[Machine monotonicity]
  \label{lem:machine_monotone}
  If $\heval{t}{\sigma}{v}{\sigma'}$, then $\sigma \heaple \sigma'$.
\end{lemma}
\begin{proof}
  Straightforward induction on $\heval{t}{\sigma}{v}{\sigma'}$.
\end{proof}

\begin{lemma}
  \label{lem:heaplePreserve}
  ~
  \begin{enumerate}[(i)]
  \item If $\sigma \tickle \sigma'$, then
    $\gc{\sigma} \heaple \gc{\sigma'}$.
  \item $\gc{\sigma} \tickle \sigma$.
  \item If $\eta \heaple \eta'$ 
  then $\mktick{\eta} \heaple \mktick{\eta'}$.
  \end{enumerate}
\end{lemma}
\begin{proof}
  By a straightforward case analysis.
\end{proof}

\begin{lemma}~
  \label{lem:heapOperations}
  \begin{enumerate}[(i)]
  \item $\restr{\mktick{\eta}} = \mktick{\restr{\eta}}$.
  \item $\restr{\gc{\sigma}} = \gc{\restr{\sigma}}$.
  \end{enumerate}
\end{lemma}
\begin{proof}
  By a straightforward case analysis.
\end{proof}

\begin{thmcpy}{Lemma~\ref{lem:vinterp_weakening}}
  Let $n \ge n'$, and $\sigma \tickle[\Delta] \sigma'$.
  \begin{enumerate}[(i)]
  \item \label{item:vinterp_weakeningI}
    $\vinterp{A}{n,\sigma} \subseteq
    \vinterp[\Delta]{A}{n',\sigma'}$.
  \item \label{item:vinterp_weakeningII}
    $\tinterp{A}{n,\sigma} \subseteq
    \tinterp[\Delta]{A}{n',\sigma'}$.
  \item \label{item:vinterp_weakeningIII}
    $\cinterp{\Gamma}{n,\sigma} \subseteq
    \cinterp[\Delta]{\Gamma}{n',\sigma'}$.
  \end{enumerate}
\end{thmcpy}
\begin{proof}[Proof of Lemma~\ref{lem:vinterp_weakening}]
  \ref{item:vinterp_weakeningI} and \ref{item:vinterp_weakeningII} are
  proved by a well-founded induction using the same well-founded order
  that we used to argue that both logical relations are
  well-defined. \ref{item:vinterp_weakeningIII} is proved by induction
  on the length of $\Gamma$, and using \ref{item:vinterp_weakeningI}.
\end{proof}

\begin{thmcpy}{Lemma~\ref{lem:vinterp_gc}}
  ~
  \begin{enumerate}[(i)]
  \item
    $\vinterp{A}{n,\sigma} \subseteq
    \vinterp{A}{n,\gc{\sigma}}$.
  \item
    $\cinterp{\Gamma}{n,\sigma} \subseteq
    \cinterp{\Gamma}{n,\gc{\sigma}}$ if $\Gamma$ is tick-free.
  \end{enumerate}
\end{thmcpy}
\begin{proof}[Proof of lemma~\ref{lem:vinterp_gc}]
  Both items are proved by induction on the size $\tsize{A}$ and on
  the length of $\Gamma$, respectively.
\end{proof}

\begin{thmcpy}{Lemma~\ref{lem:vinterp_restr}}
  ~
  \begin{enumerate}[(i)]
  \item If $v \in \vinterp{\DelayE A}{n,\sigma}$, then
    $v \in \vinterp{\DelayE A}{n,\sigma'}$, for any $\sigma'$ 
    with $\restr[\cl v]{\sigma} = \restr[\cl v]{\sigma'}$.
  \item If $\gamma \in \cinterp{\Gamma,\tick[\theta]}{n,\sigma}$
    and $\Theta = \cleval{\theta\gamma}$
    then $\gamma \in \cinterp{\Gamma,\tick[\theta]}{n,\restr\sigma}$
  \end{enumerate}
\end{thmcpy}
\begin{proof}[Proof of Lemma~\ref{lem:vinterp_restr}]
  Both items are proved by inspection of the definitions of
  $\vinterp{\DelayE A}{n,\sigma}$ and
  $\cinterp{\Gamma,\tick[\theta]}{n,\sigma}$, respectively.
\end{proof}
The fact that (i) holds for $\vinterp{\DelayE A}{n,\sigma}$ but
not for $\tinterp{\DelayE A}{n,\sigma}$ is the reason we needed
to restrict the calculus so that $\adv$ and $\select$ may only be
applied to values.

\begin{lemma}
  \label{lem:vinterp_stable}
  $\vinterp{A}{n,\sigma} = \vinterp{A}{n,\emptyset}$
  for any stable type $A$.
\end{lemma}
\begin{proof}
  By straightforward induction on the size of $A$.
\end{proof}

\begin{lemma}
  \label{lem:cinterp_tokenfree_restr}
  Let $\gamma \in \cinterp{\Gamma,\Gamma'}{n,\sigma}$ such
  that $\Gamma'$ is tick-free. Then
  $\gamma|_{\Gamma}\in \cinterp{\Gamma}{n,\sigma}$.
\end{lemma}
\begin{proof}
  By a straightforward induction on the length of $\Gamma'$.
\end{proof}

\begin{lemma}
  \label{lem:stabilize}
  If $\gamma \in \cinterp{\Gamma}{n,\sigma}$, then
  $\gamma|_{\stabilize\Gamma} \in
  \cinterp{\stabilize\Gamma}{n,\emptyset}$.
\end{lemma}
\begin{proof}
  By a straightforward induction on the length of $\Gamma$ and using
  Lemma~\ref{lem:vinterp_stable}.
\end{proof}

\begin{lemma}
  \label{lem:vinterp_value}
  If $t \in \vinterp{A}{w}$, then $t$ is a value.
\end{lemma}
\begin{proof}
  By inspection of the definition.
\end{proof}

\begin{lemma}
  \label{lem:vinterp_tinterp}
  $\vinterp{A}{w} \subseteq \tinterp{A}{w}$.
\end{lemma}
\begin{proof}
  Let $t \in \vinterp{A}{n,\sigma}$, 
  and $\sigma'\tickge[\Delta] \sigma$ and $\iota : \Delta$. By
  Lemma~\ref{lem:vinterp_value}, $t$ is a value and thus
  $\heval{t}{\sigma'}{t}{\sigma'}$. By
  Lemma~\ref{lem:vinterp_weakening}, we have that
  $t \in \vinterp{A}{n,\sigma'}$, which in turn implies
  $t \in \vinterp[\Delta]{A}{n,\sigma'}$ by
  Lemma~\ref{lem:vinterp_weakening}.
\end{proof}

\begin{lemma}
  \label{lem:tinterp_vinterp}
  If $v$ is a value with $v \in \tinterp{A}{w}$, then $v \in \vinterp{A}{w}$.
\end{lemma}
\begin{proof}
  Let $v \in \tinterp{A}{n,\sigma}$, and pick an arbitrary
  $\iota : \Delta$. Since $\heval{v}{\sigma}{v}{\sigma}$, we have by
  definition that $v \in \vinterp{A}{n,\sigma}$.
\end{proof}

\begin{lemma}
  \label{lem:clockExprSubs}
  If $\isclock{\Gamma}{\theta}$ and $\gamma \in \cinterp{\Gamma}{w}$,
  then $\theta\gamma$ is a closed clock expression and
  $\cleval{\theta\gamma} \subseteq \domp{\Delta}$.
\end{lemma}
\begin{proof}
  We proceed by induction on $\isclock{\Gamma}{\theta}$.
  \begin{itemize}
  \item $ \inferrule*%
    {\isclock{\Gamma}{\theta} \\ \isclock{\Gamma}{\theta'}}%
    {\isclock{\Gamma}{\theta \sqcup \theta'}}$

    By induction, $\theta\gamma$ and $\theta'\gamma$ are closed and
    $\cleval{\theta\gamma} \subseteq \dom{\Delta}$. Hence,
    $(\theta \sqcup \theta')\gamma = \theta\gamma \sqcup
    \theta'\gamma$ is closed and $\cleval{(\theta \sqcup
      \theta')\gamma} = \cleval{\theta\gamma} \cup
    \cleval{\theta'\gamma} \subseteq \dom{\Delta}$.
  \item $\inferrule*%
    {\hastype{\Gamma}{v}{\DelayE A}}%
    {\isclock{\Gamma}{\cl{v}}}$

    $\hastype{\Gamma}{v}{\DelayE A}$ implies that
    $v\gamma \in \vinterp{\DelayE A}{w}$ (because either $v$ is a
    variable or $v = \inp[\kappa]$ for some clock $\kappa$). Hence,
    $v\gamma \in \locs[\Delta]\cup\setcom{\inp[\kappa]}{\kappa \in
      \domp{\Delta}}$ and thus $\cl{v}\gamma$ is a closed clock
    expression and $\cleval{\cl{v}\gamma}\subseteq \domp\Delta$.
  \end{itemize}
\end{proof}

\begin{lemma}
  \label{lem:vinterpValueType}
  Let $A$ be a value type. Then $v \in \vinterp{A}{n,\sigma}$ iff $\hastype{}{v}{A}$.
\end{lemma}
\begin{proof}
  Straightforward induction on $A$.
\end{proof}

\begin{lemma}
  \label{lem:adv}
  Let $\eta_N \in \heaps[\kappa]$,  
  $v \in \vinterp{\DelayE
    A}{n+1,(\eta_N,\clnin{\eta_L})}$,
  $\iota : \Delta$, $\hastype[]{}{w}{\Delta(\kappa)}$ and
  $\kappa \in \cleval{\cl{v}}$. Then, for any
  $\sigma \heapge \eta_N\stick{\kappa}{w}\eta_L$, there are some
  $\sigma'$ and $v' \in \vinterp{A}{n,\sigma'}$ with
  $\heval{\adv\,v}{\sigma}{v'}{\sigma'}$.
\end{lemma}

\begin{proof}
    By definition, $v$ is either some $l \in \locs$ or of the
    form $\inp[\kappa']$.
    \begin{itemize}
    \item In the former case, we have by the definition of the value
      relation and Lemma~\ref{lem:heapOperations},
      \[(\eta_N,\clnin{\eta_L})(l) \in
      \tinterp{A}{n,\restr[\cl{l}]{\eta_N\stick{\kappa}{w}\clnin{\eta_L}}}.\]
      In turn, this implies by Lemma~\ref{lem:vinterp_weakening} that
      \[(\eta_N,\clnin{\eta_L})(l) \in
      \tinterp{A}{n,\eta_N\stick{\kappa}{w}\eta_L}.\] Moreover,
      since $\kappa \in \cl l$ we know that
      $(\eta_N,\clnin{\eta_L})(l) = \eta_N(l)$. Hence, there is a
      reduction $\heval{\eta_N(l)}{\sigma}{v'}{\sigma'}$ with
      $v' \in \vinterp{A}{n,\sigma'}$, which by definition
      means that $\heval{\adv\,v}{\sigma}{v'}{\sigma'}$.
    \item In the latter case, we know that $\kappa' = \kappa$ because
      $\kappa \in \cleval{\cl{\inp[\kappa']}} =
      \set{\kappa'}$. Moreover, we have that $\kappa :_\capa A \in \Delta$
      for $\capa\in\{\cpush,\cbp\}$
      and thus $\hastype[]{}{w}{A}$. By definition,
      $\eta_N\stick{\kappa}{w}\eta_L \heaple \sigma$ implies that
      $\sigma$ is of the form
      $\eta'_N\stick{\kappa}{w}\eta'_L$. Hence, by definition
      $\heval{\adv\,\inp}{\sigma}{w}{\sigma}$. Moreover, by
      Lemma~\ref{lem:vinterpValueType} 
      $w \in \vinterp{A}{n,\sigma}$.
    \end{itemize}

\end{proof}
\begin{thmcpy*}{Theorem~\ref{thr:lrl}}{Fundamental property}
  Given $\wfcxt{\Gamma}$, $\hastype{\Gamma}{t}{A}$, and $\gamma \in \cinterp{\Gamma}{n,\sigma}$,
  then $t\gamma \in \tinterp{A}{n,\sigma}$.
\end{thmcpy*}
\begin{proof}[Proof of Theorem~\ref{thr:lrl}]
  We proceed by structural induction over the typing derivation
  $\hastype{\Gamma}{t}{A}$. If $t\gamma$ is a value, it suffices to
  show that $t\gamma \in \vinterp{A}{n,\sigma}$, according to
  Lemma~\ref{lem:vinterp_tinterp}. In all other cases, to prove
  $t\gamma \in \tinterp{A}{n,\sigma}$, we assume some
  input buffer $\iota : \Delta$ and store
  $\sigma' \tickge[\Delta] \sigma$, and show that there exists
  $\sigma''$ and $v$ s.t. $\heval{t\gamma}{\sigma}{v}{\sigma''}$ and
  $v \in \vinterp[\Delta]{A}{n,\sigma''}$. By
  Lemma~\ref{lem:vinterp_weakening} we may assume that
  $\gamma \in \cinterp[\Delta]{\Gamma}{n,\sigma'}$.
  
  \begin{itemize}
  \item $  \inferrule*%
  {\Gamma' \text{ tick-free or } A \; \stable \\ \wfcxt{\Gamma,x:A,\Gamma'}}%
  {\hastype{\Gamma,x:A,\Gamma'}{x}{A}}$\\

    We show that $x\gamma \in \vinterp{A}{n,\sigma}$.  If
    $\Gamma'$ is tick-free, then
    $x\gamma \in \vinterp{A}{n,\sigma}$ by
    Lemma~\ref{lem:cinterp_tokenfree_restr}

    If $\Gamma'$ is not tick-free, it is of the form
    $\Gamma_1,\tick[\theta],\Gamma_2$ and $A$ is stable. By
    Lemma~\ref{lem:cinterp_tokenfree_restr},
    \[\gamma|_{\Gamma,x:A} \in \cinterp{\Gamma, x : A}{n+1,\sigma'}\]
    for some $\sigma'$. Hence,
    $x\gamma \in \vinterp{A}{n+1,\sigma'}$ and by
    Lemma~\ref{lem:vinterp_weakening} and
    Lemma~\ref{lem:vinterp_stable}
    $x\gamma \in \vinterp{A}{n,\sigma}$.

  \item $\inferrule* {~}%
    {\hastype{\Gamma}{\unit}{\Unit}}$\\
    
    Follows immediately by definition.

  \item $  \inferrule*%
    {\hastype{\Gamma}{s}{A} \\ \hastype{\Gamma,x:A}{t}{B}}%
    {\hastype{\Gamma}{\letterm x s t}{B}}$\\

    By induction, we have $s\gamma \in \tinterp{A}{n,\sigma}$, which
    means that $\heval{s\gamma}{\sigma'}{v}{\sigma''}$ for some
    $v \in \vinterp[\Delta]{A}{n, \sigma''}$. By
    Lemma~\ref{lem:vinterp_weakening} and
    Lemma~\ref{lem:machine_monotone},
    $\gamma \in \cinterp{\Gamma}{n,\sigma''}$ and thus
    \[\gamma[x \mapsto v] \in \cinterp[\Delta]{\Gamma,x :
      A}{n,\sigma''}.\] Hence, we may apply the induction hypothesis to obtain 
    $t\gamma[x \mapsto v] \in \tinterp{B}{n,\sigma''}$.  Since all
    elements in the range of $\gamma$ are closed terms,
    $t\gamma[x \mapsto v] = (t\gamma)[v/x]$ and thus
    $(t\gamma)[v/x] \in
    \tinterp[\Delta]{B}{n,\sigma''}$. Consequently,
    $\heval{(t\gamma)[v/x]}{\sigma''}{w}{\sigma'''}$ with
    $w \in \vinterp{B}{n,\sigma'''}$. By definition, we thus have
    $\heval{(\letterm x s t)\gamma}{\sigma'}{w}{\sigma'''}$ with
    $w \in \vinterp{B}{n,\sigma'''}$.
  \item $\inferrule*%
    {\hastype{\Gamma,x:A}{t}{B} \\ \Gamma\text{ tick-free}}%
    {\hastype{\Gamma}{\lambda x.t}{A \to B}}$

    We show that
    $\lambda x. t\gamma \in \vinterp{A \to B}{n,\sigma}$. To
    this end, we assume 
    $\sigma' \tickge[\Delta] \gc{\sigma}$,
    $n' \le n$, and $v \in \vinterp[\Delta]{A}{n',\sigma'}$,
    with the goal of showing
    $(t\gamma)[v/x] \in \tinterp[\Delta]{B}{n',\sigma'}$. By
    Lemma~\ref{lem:vinterp_gc} and Lemma~\ref{lem:vinterp_weakening},
    $\gamma \in \cinterp[\Delta]{\Gamma}{n',\sigma'}$, and
    thus, by definition,
    $\gamma[x \mapsto v] \in \cinterp[\Delta]{\Gamma,x :
      A}{n',\sigma'}$. By 
    induction, we then have that
    $t\gamma[x \mapsto v] \in \tinterp[\Delta]{B}{n',\sigma'}$. Since
    all elements in the range of $\gamma$ are closed terms,
    $t\gamma[x \mapsto v] = (t\gamma)[v/x]$ and thus
    $(t\gamma)[v/x] \in \tinterp[\Delta]{B}{n',\sigma'}$.
  \item $\inferrule*%
    {\hastype{\Gamma}{s}{A \to B} \\
      \hastype{\Gamma}{t}{A}}%
    {\hastype{\Gamma}{s\,t}{B}}$

    By induction, we have $s\gamma \in \tinterp{A \to B}{n,\sigma}$,
    which means that
    $\heval{s\gamma}{\sigma'}{\lambda x. s'}{\sigma''}$ for some
    $\lambda x. s' \in \vinterp[\Delta]{A \to B}{n, \sigma''}$.  By
    induction, we also have $t\gamma \in \tinterp{A}{n,\sigma}$. Since
    by Lemma~\ref{lem:machine_monotone},
    $\sigma'' \tickge[\Delta] \sigma$, this means that
    $\heval{t\gamma}{\sigma''}{v}{\sigma'''}$ for some
    $v \in \vinterp[\Delta]{A}{n, \sigma'''}$. Hence, by definition,
    $s'[v/x] \in \tinterp[\Delta]{B}{n,\sigma'''}$, since
    $\sigma''' \tickge[\Delta'] \gc{\sigma''}$ by
    Lemma~\ref{lem:heaplePreserve} and
    Lemma~\ref{lem:machine_monotone}. That means that we have
    $\heval{s[v/x]}{\sigma'''}{w}{\sigma''''}$ for some
    $w \in \vinterp[\Delta]{B}{n,\sigma''''}$. By definition of the
    machine, we thus have
    $\heval{(s\gamma)(t\gamma)}{\sigma'}{w}{\sigma''''}$.
  \item $\inferrule*%
    {\hastype{\Gamma}{t}{A} \\
      \hastype{\Gamma}{t'}{B}}%
    {\hastype{\Gamma}{\pair{t}{t'}}{A \times B}}$\\

    By induction, we have $s\gamma \in \tinterp{A}{n,\sigma}$, which
    means that $\heval{s\gamma}{\sigma'}{v}{\sigma''}$ for some
    $v \in \vinterp[\Delta]{A}{n, \sigma''}$. We also have
    $t\gamma \in \tinterp{A}{n,\sigma}$ by induction, which by
    Lemma~\ref{lem:machine_monotone} means that
    $\heval{t\gamma}{\sigma''}{v'}{\sigma'''}$ for some
    $v' \in \vinterp[\Delta]{B}{n, \sigma'''}$. Hence,
    $\heval{\pair{t}{t'}}{\sigma'}{\pair{v}{v'}}{\sigma'''}$, and by
    Lemma~\ref{lem:vinterp_weakening}, $\pair{v}{v'} \in
    \vinterp[\Delta]{A \times B}{n, \sigma'''}$.

  \item $\inferrule*%
    {\hastype{\Gamma}{t}{A_1 \times A_2} \\ i \in \{1, 2\}}%
    {\hastype{\Gamma}{\pi_i\,t}{A_i}}$\\

    By induction, we have
    $t\gamma \in \tinterp{A_1\times A_2}{n,\sigma}$, which means that
    $\heval{t\gamma}{\sigma'}{\pair{v_1}{v_2}}{\sigma''}$ with
    $v_i \in \vinterp[\Delta]{A_i}{n, \sigma''}$. Moreover, by
    definition, $\heval{\pi_i\,t\gamma}{\sigma'}{v_i}{\sigma''}$.
    
  \item $\inferrule*%
    {\hastype{\Gamma}{t}{A_i} \\ i \in \{1, 2\}}%
    {\hastype{\Gamma}{\interm_i\, t}{A_1 + A_2}}$

    By induction, we have $t\gamma \in \tinterp{A_i}{n,\sigma}$, which
    means that $\heval{t\gamma}{\sigma'}{v}{\sigma''}$ with
    $v \in \vinterp[\Delta]{A_i}{n, \sigma''}$. Hence, by definition,
    $\heval{\interm_i\,t\gamma}{\sigma'}{\interm_i\,v}{\sigma''}$ and
    $\interm_i\,v \in \vinterp[\Delta]{A_1 + A_2}{n, \sigma''}$.

  \item $\inferrule*%
    {\hastype{\Gamma,x: A_i}{t_i}{B} \\
      \hastype{\Gamma}{t}{A_1 + A_2} \\ i \in \{1,2\}}%
    {\hastype{\Gamma}{\caseterm {t}{x}{t_1}{x}{t_2}}{B}}$\\

    By induction, we have $t\gamma \in \tinterp{A_1 + A_2}{n,\sigma}$,
    which means that $\heval{t\gamma}{\sigma'}{\interm_i\, v}{\sigma''}$
    for some $i \in \set{1,2}$ such that
    $v \in \vinterp{A_i}{n, \sigma''}$. By
    Lemma~\ref{lem:vinterp_weakening} and
    Lemma~\ref{lem:machine_monotone},
    $\gamma \in \cinterp{\Gamma}{n,\sigma''}$ and thus
    $\gamma[x \mapsto v] \in \cinterp[\Delta]{\Gamma,x :
      A_i}{n,\sigma''}$. Hence, we may apply the induction hypothesis
    to obtain $t_i\gamma[x \mapsto v] \in \tinterp{B}{n,\sigma''}$.
    Since all elements in the range of $\gamma$ are closed terms,
    $t_i\gamma[x \mapsto v] = (t_i\gamma)[v/x]$ and thus
    $(t_i\gamma)[v/x] \in \tinterp[\Delta]{B}{n,\sigma''}$.
    Consequently, $\heval{(t_i\gamma)[v/x]}{\sigma''}{w}{\sigma'''}$
    with $w \in \vinterp{B}{n,\sigma'''}$. By definition, we thus have
    $\heval{(\caseterm
      {t}{x}{t_1}{x}{t_2})\gamma}{\sigma'}{w}{\sigma'''}$, as well.

  \item $\inferrule*%
    {\hastype{\Gamma,\tick[\theta]}{t}{A} \\ \isclock{\Gamma}{\theta}}%
    {\hastype{\Gamma}{\delay_\theta\,t}{\DelayE A}}$

    By definition of the machine we have that
    $\heval{\delay_{\theta\gamma}\,t\gamma}{\sigma'}{l}{\sigma''}$,
    where $\sigma''=\sigma',l\mapsto t\gamma$ and
    $\cl l = \cleval{\theta\gamma}$. By Lemma~\ref{lem:clockExprSubs},
    $\cleval{\theta\gamma} \subseteq \domp{\Delta}$. It remains to be
    shown that
    $l \in \vinterp[\Delta]{\DelayE A}{n, \sigma''}$. For the
    case where $n=0$, this follows immediately from the fact that
    $\cleval{\theta\gamma} \subseteq \domp{\Delta}$.

    Assume that $n = n' + 1$, $\kappa \in \Theta$, and
    $\hastype[]{}{v}{\Delta(\kappa)}$, where
    $\Theta = \cleval{\theta\gamma}$. By
    Lemma~\ref{lem:vinterp_weakening} and Lemma~\ref{lem:vinterp_gc},
    we have that
    $\gamma \in \cinterp[\Delta]{\Gamma}{n'+1, \gc{\sigma''}}$. By
    definition we thus have that
    \[\gamma \in \cinterp[\Delta]{\Gamma,\tick[\theta]}{n',
      \clin{\gc{\sigma''}}\stick\kappa v \clnin{\gc{\sigma''}}} = 
      \cinterp[\Delta]{\Gamma,\tick[\theta]}{n', \mktick{\gc{\sigma''}}},\] and thus
    $\gamma \in \cinterp[\Delta]{\Gamma,\tick[\theta]}{n',
      \restr{\mktick{\gc{\sigma''}}}}$
    according to Lemma~\ref{lem:vinterp_restr}. Hence, we can apply
    the induction hypothesis to conclude that
    \[t\gamma \in \tinterp[\Delta]{A}{n',
        \restr{\mktick{\gc{\sigma''}}}}.\]
    Since $\sigma''(l) = t\gamma$, we thus have that
    $l \in \vinterp[\Delta]{\DelayE A}{n, \sigma''}$.

  \item $\inferrule*%
    {~}%
    {\hastype{\Gamma}{\never}{\DelayE A}}$

    According to the definition of the machine, we have
    $\heval{\never}{\sigma'}{l}{\sigma'}$ with
    $l = \allocate[\emptyset]{\sigma}$. Since $\cl l = \emptyset$, we
    know that $l \in \vinterp[\Delta]{\DelayE A}{n,\sigma'}$.

  \item $ \inferrule*%
    {\kappa :_\capa A \in \Delta, \capa \in \{\cpush, \cbp\}}%
    {\hastype{\Gamma}{\inp}{\DelayE A}}%
    $\\
    $\inp\gamma = \inp \in \vinterp{A}{n,\sigma}$ follows
    immediately by definition and the premise.
  \item $ \inferrule*%
    {\kappa :_\capa A \in \Delta, \capa \in \{\cbuf, \cbp\}}%
    {\hastype{\Gamma}{\buf}{A}}%
    $\\
    Since $\iota : \Delta$, we
    know that $\hastype[]{}{\iota(\kappa)}{A}$. Hence, By definition
    of the machine
    $\heval{\buf}{\sigma'}{\iota(\kappa)}{\sigma'}$. Moreover, by
    Lemma~\ref{lem:vinterpValueType} 
    $\iota(\kappa) \in \vinterp[\Delta]{A}{n,\sigma'}$.
  \item $  \inferrule*%
    {\hastype{\Gamma}{v}{\DelayE A}\\ \wfcxt{\Gamma,\tick[\cl{v}],\Gamma'}}%
    {\hastype{\Gamma,\tick[\cl{v}],\Gamma'}{\adv\,v}{A}}$
  
    By Lemma~\ref{lem:vinterp_weakening} we have that
    $\gamma \in \cinterp[\Delta]{\Gamma,\tick[\cl
      v],\Gamma'}{n,\sigma'}$ and by
    Lemma~\ref{lem:cinterp_tokenfree_restr},
    $\gamma|_\Gamma \in \cinterp[\Delta]{\Gamma,\tick[\cl
      v]}{n,\sigma'}$. Let
    $\Theta = \cleval{\cl{v}\gamma|_\Gamma}$. By definition of the
    context relation, $\Theta$ is well-defined and a subset of
    $\domp{\Delta}$. By definition of the context relation we also
    find $\kappa \in \Theta$, $\eta_N, \eta_L$,
    $\eta'_N, \eta'_L$ such that
    $\sigma' = \eta_N\stick{\kappa}{w}\eta_L$,
    $\hastype[]{}{w}{\Delta(\kappa)}$,
    $\restr{\eta_N} = \restr{\eta'_N}$,
    $\restr{\eta_L} = \restr{\eta'_L}$,
    and
    $\gamma|_\Gamma \in
    \cinterp[\Delta]{\Gamma}{n+1,(\eta'_N,\clnin{\eta'_L})}$.
    By induction, we thus have that
    $v\gamma \in \tinterp[\Delta]{\DelayE
      A}{n+1,(\eta'_N,\clnin{\eta'_L})}$.
    Since $v\gamma$ is a value we have
    $v\gamma \in \vinterp[\Delta]{\DelayE
      A}{n+1,(\eta'_N,\clnin{\eta'_L})}$ by
    Lemma~\ref{lem:tinterp_vinterp}.  By Lemma~\ref{lem:vinterp_restr}
    we then have that
    \[
      v\gamma \in \vinterp[\Delta]{\DelayE
        A}{n+1,(\eta_N,\clnin{\eta_L})}.
    \]
    By Lemma~\ref{lem:adv}, we then find a reduction
    $\heval{\adv\,v\gamma}{\sigma'}{v'}{\sigma''}$ with
    $v' \in \vinterp[\Delta]{A}{n,\sigma''}$.
    
  \item $\inferrule*%
    {\hastype{\Gamma}{v_1}{\DelayE A_1}
      \\\hastype{\Gamma}{v_2}{\DelayE A_2}\\
      \vdash \theta = \cl{v_1} \sqcup \cl{v_2}\\
      \wfcxt{\Gamma,\tick[\theta],\Gamma'}}%
    {\hastype{\Gamma,\tick[\theta],\Gamma'}{\select\,v_1\,v_2}{((A_1
        \times \DelayE A_2) + (\DelayE A_1
        \times A_2)) + (A_1 \times A_2)}}$

    By Lemma~\ref{lem:vinterp_weakening} we have that
    $\gamma \in
    \cinterp[\Delta]{\Gamma,\tick[\theta],\Gamma'}{n,\sigma'}$
    and by Lemma~\ref{lem:cinterp_tokenfree_restr},
    $\gamma|_\Gamma \in
    \cinterp[\Delta]{\Gamma,\tick[\theta]}{n,\sigma'}$. Let
    $\Theta_1 = \cleval{\cl{v_1}\gamma|_\Gamma}$,
    $\Theta_2 = \cleval{\cl{v_2}\gamma|_\Gamma}$, and
    $\Theta = \Theta_1 \cup \Theta_2$. According to the definition of
    the context relation, $\Theta$ is well-defined and a subset of
    $\dom{\Delta}$. By definition of the context relation we also
    find $\kappa \in \Theta$, $\eta_N, \eta_L$,
    $\eta'_N, \eta'_L$ such that
    $\sigma' = \eta_N\stick{\kappa}{w}\eta_L$,
    $\hastype[]{}{w}{\Delta(\kappa)}$,
    $\restr{\eta_N} = \restr{\eta'_N}$,
    $\restr{\eta_L} = \restr{\eta'_L}$,
    and $\gamma|_\Gamma \in
    \cinterp[\Delta]{\Gamma}{n+1,(\eta'_N,\clnin{\eta'_L})}$.
    By induction hypothesis, we thus
    have that
    $v_i\gamma \in \tinterp[\Delta]{\DelayE
      A_i}{n+1,(\eta'_N,\clnin{\eta'_L})}$
    for all $i \in \set{1,2}$. Since $v_i\gamma$ are values, we also
    have that
    $v_i\gamma \in \vinterp[\Delta]{\DelayE
      A_i}{n+1,(\eta'_N,\clnin{\eta'_L})}$
    by Lemma~\ref{lem:tinterp_vinterp}. By
    Lemma~\ref{lem:vinterp_restr} 
    we then have that
    $v_i\gamma \in \vinterp[\Delta]{\DelayE
      A_i}{n+1,(\eta_N,\clnin{\eta_L})}$ for
    all $i \in \set{1,2}$.  There are two cases to consider:
    \begin{itemize}
    \item Let $i \in \set{1,2}$ and $j = 3-i$ such that
      $\kappa \in \Theta_i\setminus\Theta_j$:

      By Lemma~\ref{lem:adv}, there is a reduction
      $\heval{\adv\,v_i\gamma}{\sigma'}{u_i}{\sigma''}$ with
      $u_i \in \vinterp[\Delta]{A_i}{n,\sigma''}$, which by definition
      means that
      $\heval{\select\,v_1\gamma\,v_2\gamma}{\sigma'}{\interm_1(\interm_i\pair{u_1}{u_2})}{\sigma''}$
      with $u_j = v_j\gamma$. It thus remains to be shown that
      $v_j\gamma \in \vinterp[\Delta]{\DelayE A_j}{n,\sigma''}$. There
      are two cases to consider.
      \begin{itemize}
      \item Let $v_j\gamma = l$ for some $l \in \locs[\Delta]$. From
        $l \in \vinterp[\Delta]{\DelayE
          A_j}{n+1,(\eta_N,\clnin{\eta_L})}$
        and the fact that $\kappa \nin \Theta_j$, we obtain that
        $l \in \vinterp[\Delta]{\DelayE A_j}{n+1,\clnin{\eta_L}}$ by
        using Lemma~\ref{lem:vinterp_restr}. In particular, we use the fact
        that $\restr[\Theta_j]{\eta_N} = \emptyset$ since
        $\eta_N \in \heaps[\kappa]$ and $\kappa\nin\Theta_j$. Since
        $\clnin{\eta_L} \tickle[\Delta] \sigma'$ and, by
        Lemma~\ref{lem:machine_monotone}, $\sigma' \tickle[\Delta] \sigma''$,
        we can then use Lemma~\ref{lem:vinterp_weakening} to conclude
        that $l \in \vinterp[\Delta]{\DelayE A_j}{n,\sigma''}$.
      \item Let $v_j\gamma = \inp[\kappa']$ for some clock
        $\kappa'$. But then $\hastype{\Gamma}{v_j}{\DelayE A_j}$ is
        due to $\kappa' :_{\cpush} A_j \in \Delta$ or $\kappa' :_{\cbp} A_j \in \Delta$ and thus
        $v_j\gamma \in \vinterp[\Delta]{\DelayE A_j}{n,\sigma''}$
        follows immediately by definition of the value relation.
      \end{itemize}

    \item $\kappa \in \Theta_1 \cap \Theta_2$:
      By Lemma~\ref{lem:adv}, we obtain a
      reduction $\heval{\adv\,v_1\gamma)}{\sigma'}{v'_1}{\sigma''}$ with
      $v'_1 \in \vinterp[\Delta]{A_1}{n,\sigma''}$, and 
      a reduction $\heval{\adv\,v_2\gamma}{\sigma''}{v'_2}{\sigma'''}$ with
      $v'_2 \in \vinterp[\Delta]{A_2}{n,\sigma'''}$. By definition we thus
      obtain a reduction
      $\heval{\select\,(v_1\gamma)\,(v_2\gamma)}{\sigma'}{\interm_2(\pair{v'_1}{v'_2})}{\sigma'''}$. Moreover,
      applying Lemma~\ref{lem:machine_monotone} and
      Lemma~\ref{lem:vinterp_weakening}, we obtain that
      $v'_1 \in \vinterp[\Delta]{A_1}{n,\sigma'''}$, which means that we
      have
      $\interm_2(\pair{v'_1}{v'_2}) \in \vinterp[\Delta]{A_1 \times
        A_2}{n,\sigma'''}$.
    \end{itemize}

  \item $\inferrule* {~}%
    {\hastype{\Gamma}{\zero}{\Nat}}$

    $\zero\gamma \in \vinterp{\Nat}{n,\sigma}$ follows
    immediately by definition.
    
  \item $\inferrule*%
    {\hastype{\Gamma}{t}{\Nat}}%
    {\hastype{\Gamma}{\suc \, t}{\Nat}}$

    By induction hypothesis
    $t\gamma \in \tinterp{\Nat}{n,\sigma}$, which means that
    $\heval{t\gamma}{\sigma'}{\suc^m\,\zero}{\sigma''}$ for some
    $m\in\nats$.  Hence, by definition,
    $\heval{\suc\,t\gamma}{\sigma'}{\suc^{m+1}\,\zero}{\sigma''}$ and
    $\suc^{m+1}\,\zero \in \vinterp[\Delta]{\Nat}{n,\sigma''}$.

  \item $\inferrule*%
    {\hastype{\Gamma}{s}{A} \\
      \hastype{\Gamma,x:\Nat,y:A}{t}{A} \\ \hastype{\Gamma}{u}{\Nat}}
    {\hastype{\Gamma}{\recN{s}{x}{y}{t}{u}}{A}}$

    We claim that the following holds:
    \begin{equation}
      \label{eq:recN}
      \recN{s\gamma}{x}{y}{t\gamma}{\suc^k\,\zero} \in
      \tinterp{A}{n,\sigma} \text{ for all } k \in \nats
    \end{equation}
    
    To show that
    $\recN{s\gamma}{x}{y}{t\gamma}{u\gamma} \in
    \tinterp{A}{n,\sigma}$, assume some
    $\sigma'\tickge[\Delta] \sigma$ and
    $\iota : \Delta$. By induction hypothesis
    $u\gamma \in \tinterp[\Delta]{\Nat}{n,\sigma}$, which means that
    $\heval{u\gamma}{\sigma'}{\suc^m\,\zero}{\sigma''}$. By
    \eqref{eq:recN} and Lemma~\ref{lem:machine_monotone} we have that
    $\heval{\recN{s\gamma}{x}{y}{t\gamma}{\suc^m\,\zero}}{\sigma''}{v}{\sigma'''}$
    with $v \in\tinterp[\Delta]{A}{n,\sigma'''}$. By definition of the
    machine, we also have that
    $\heval{\recN{s\gamma}{x}{y}{t\gamma}{u\gamma}}{\sigma'}{v}{\sigma'''}$.

    We conclude by showing \eqref{eq:recN} by induction on $k$.
    \begin{itemize}
    \item Case $k = 0$: Let 
      $\sigma'\tickge[\Delta] \sigma$ and $\iota : \Delta$. By
      definition, $\heval{0}{\sigma'}{0}{\sigma'}$. By induction
      hypothesis $s\gamma \in \tinterp[\Delta]{A}{n,\sigma}$,
      which means that $\heval{s\gamma}{\sigma'}{v}{\sigma''}$ for
      some $v \in \vinterp[\Delta]{A}{n,\sigma''}$. By
      definition,
      $\heval{\recN{s\gamma}{x}{y}{t\gamma}{0}}{\sigma'}{v}{\sigma''}$
      follows.
    \item Case $k = l + 1$. Let 
      $\sigma'\tickge[\Delta] \sigma$ and $\iota : \Delta$. By
      definition,
      $\heval{\suc^{k}\,0}{\sigma'}{\suc(\suc^l\,\zero)}{\sigma'}$. By
      induction (on $k$), we have that
      $\heval{\recN{s\gamma}{x}{y}{t\gamma}{\suc^l\,\zero}}{\sigma'}{v}{\sigma''}$
      with $v \in \vinterp[\Delta]{A}{n,\sigma''}$. By
      Lemma~\ref{lem:machine_monotone} and
      Lemma~\ref{lem:vinterp_weakening}, we have
      $\gamma \in \cinterp[\Delta]{\Gamma}{n,\sigma''}$ and
      thus
      \[\gamma[x\mapsto \suc^l\,\zero,y\mapsto v] \in
      \cinterp[\Delta]{\Gamma,x:\Nat,y:A}{n,\sigma''}.\] By
      induction we thus obtain that
      \[(t\gamma)[\suc^l\,\zero/x,v/y] = t\gamma[x\mapsto
      \suc^l\,\zero,y\mapsto v] \in
      \tinterp[\Delta]{A}{n,\sigma''},\] which means that there
      is a reduction
      $\heval{(t\gamma)[\suc^l\,\zero/x,v/y]}{\sigma''}{w}{\sigma'''}$
      with $w \in \vinterp[\Delta]{A}{n,\sigma'''}$. According
      to the definition of the machine, we thus have
      \[\heval{\recN{s\gamma}{x}{y}{t\gamma}{\suc^k\,\zero}}{\sigma'}{w}{\sigma'''}.\]
    \end{itemize}
    
  \item $\inferrule*
    {\hastype{\stabilize\Gamma,x : \DelayA A}{t}{A}}
    {\hastype{\Gamma}{\fix\,x. t}{A}}$

    We will show that
    \begin{align}
      \dfix\,x.t\gamma \in \vinterp{\DelayA A}{m,\emptyset} \text{
        for
        all $m \le n$.}\label{eq:dfix}
    \end{align}

    Using \eqref{eq:dfix}, Lemma~\ref{lem:stabilize}, and
    Lemma~\ref{lem:vinterp_weakening}, we then obtain that
    $(\gamma|_{\stabilize\Gamma})[x\mapsto\dfix\,x.t\gamma] \in
    \cinterp{\stabilize\Gamma,x : \DelayA A}{n,\sigma}$. Hence,
    by induction,
    $t[\dfix\,x.t/x]\gamma =
    t(\gamma|_{\stabilize\Gamma})[x\mapsto\dfix\,x.t\gamma] \in
    \tinterp{A}{n,\sigma}$, which means that we find
    $\heval{t[\dfix\,x.t/x]\gamma}{\sigma'}{v}{\sigma''}$ with
    $v \in \vinterp[\Delta]{A}{n,\sigma''}$. By definition of the
    machine we thus obtain the desired
    $\heval{\fix\,x.t\gamma}{\sigma'}{v}{\sigma''}$.

    We prove \eqref{eq:dfix} by induction on $m$.

    If $m=0$, then \eqref{eq:dfix} follows immediately from the fact
    that $\dfix\,x.t\gamma$ is a closed term.

    Let $m= m' +1$. By Lemma~\ref{lem:vinterp_weakening} and
    Lemma~\ref{lem:stabilize},
    $\gamma|_{\stabilize\Gamma} \in
    \cinterp[\Delta]{\stabilize\Gamma}{m',\emptyset}$. By the
    induction hypothesis (on \eqref{eq:dfix}) we have
    $\dfix\,x.t\gamma \in \vinterp[\Delta]{\DelayA A}{m',\emptyset}$
    and
    thus \[(\gamma|_{\stabilize\Gamma})[x\mapsto\dfix\,x.t\gamma] \in
    \cinterp[\Delta]{\stabilize\Gamma,x : \DelayA
      A}{m',\emptyset}.\] Hence, by induction,
    $t[\dfix\,x.t/x]\gamma =
    t(\gamma|_{\stabilize\Gamma})[x\mapsto\dfix\,x.t\gamma] \in
    \tinterp[\Delta]{A}{m',\emptyset}$, which allows us to conclude
    that
    $\dfix\,x.t\gamma \in \vinterp[\Delta]{\DelayA A}{m,\emptyset}$.
    
  \item $\inferrule*%
    {\hastype{\Gamma}{x}{\DelayA A}\\ \wfcxt{\Gamma,\tick[\theta],\Gamma'}}%
    {\hastype{\Gamma,\tick[\theta],\Gamma'}{\adv\,x}{A}}$

    By Lemma~\ref{lem:vinterp_weakening} we have that
    $\gamma \in
    \cinterp[\Delta]{\Gamma,\tick[\theta],\Gamma'}{n,\sigma'}$ and by
    Lemma~\ref{lem:cinterp_tokenfree_restr},
    $\gamma|_\Gamma \in
    \cinterp[\Delta]{\Gamma,\tick[\theta]}{n,\sigma'}$. Let
    $\Theta = \cleval{\theta\gamma|_\Gamma}$. According the definition
    of the context relation, $\Theta$ is well-defined and we find
    $\kappa \in \Theta$, $\hastype[]{}{w}{\Delta(\kappa)}$,
    $\eta_N, \eta_L$, $\eta'_N, \eta'_L$ such that
    $\sigma' = \eta_N\stick\kappa w\eta_L$,
    $\restr{\eta_N} = \restr{\eta'_N}$,
    $\restr{\eta_L} = \restr{\eta'_L}$,
    and
    $\gamma|_\Gamma \in
    \cinterp[\Delta']{\Gamma}{n+1,(\eta'_N,\clnin{\eta'_L})}$.  By
    Lemma~\ref{lem:cinterp_tokenfree_restr}, we thus have that
    $\gamma(x) = \dfix\,y.t$ with
    $\dfix\,y.t \in \vinterp{\DelayA
      A}{n+1,(\eta'_N,\clnin{\eta'_L})}$. By definition, this implies
    that $t[\dfix\,y.t/y]\in \tinterp[\Delta]{A}{n,\emptyset}$. That
    is, we find a reduction
    $\heval{t[\dfix\,y.t/y]}{\sigma'}{v}{\sigma''}$ with
    $v \in \vinterp[\Delta]{A}{n,\sigma''}$, which by definition means
    that we also have a reduction
    $\heval{\adv\,x\gamma}{\sigma'}{v}{\sigma''}$.
    
  \item $\inferrule*
    {\hastype{\stabilize\Gamma}{t}{ A}}
    {\hastype{\Gamma}{\rbox\, t}{\Box A}}$

    We show that
    $\rbox\, t \gamma \in \vinterp{\Box A}{n,\sigma}$.  By
    Lemma~\ref{lem:stabilize},
    $\gamma|_{\stabilize\Gamma} \in
    \cinterp{\stabilize\Gamma}{n,\emptyset}$. Hence, by
    induction,
    $t\gamma = t\gamma|_{\stabilize\Gamma}\in
    \tinterp{A}{n,\emptyset}$, and thus
    $\rbox\,t\gamma \in \vinterp{\Box A}{n,\sigma}$.
  \item $\inferrule*
    {\hastype{\Gamma}{t}{\Box A}}
    {\hastype{\Gamma}{\unbox\, t}{A}}$

    By induction hypothesis, we have that
    $t\gamma \in \tinterp{\Box A}{n, \sigma}$. That is,
    $\heval{t\gamma}{\sigma'}{\rbox\, s}{\sigma''}$ for some
    $s \in \tinterp[\Delta]{A}{n,\emptyset}$. Hence,
    $\heval{s}{\sigma''}{v}{\sigma'''}$ such that
    $v \in \vinterp[\Delta]{A}{n,\sigma'''}$ which, by
    Lemma~\ref{lem:vinterp_weakening}, implies
    $v \in \vinterp[\Delta]{A}{n,\sigma'''}$. Moreover, by definition
    of the machine we have that
    $\heval{\unbox\,t}{\sigma'}{v}{\sigma'''}$.
  \item $\inferrule*
    {\hastype{\Gamma}{t}{\Fix \;\alpha.A}}
    {\hastype{\Gamma}{\out\,t}{A[\DelayE(\Fix \; \alpha.A)/\alpha]}}$

    By induction hypothesis
    $t\gamma \in \tinterp{\Fix\,\alpha.A}{n,\sigma}$, which
    means that $\heval{t\gamma}{\sigma'}{\into\,v}{\sigma''}$ for some
    $v\in
    \vinterp[\Delta]{A[\DelayE(\Fix\,\alpha.A)/\alpha]}{n,\sigma''}$. Moreover,
    by definition of the machine we consequently have
    $\heval{\out\, t\gamma}{\sigma'}{v}{\sigma''}$.
  \item $\inferrule*
    {\hastype{\Gamma}{t}{A[\DelayE(\Fix\,\alpha.A)/\alpha]}}
    {\hastype{\Gamma}{\into\,t}{\Fix\,\alpha.A}}$

    By induction hypothesis
    $t\gamma \in
    \tinterp{A[\DelayE(\Fix\,\alpha.A)/\alpha]}{n,\sigma}$,
    which means that $\heval{t\gamma}{\sigma'}{v}{\sigma''}$ with
    $v\in
    \vinterp[\Delta]{A[\DelayE(\Fix\,\alpha.A)/\alpha]}{n,\sigma''}$.
    Hence, by definition,
    $\heval{\into\,t\gamma}{\sigma'}{\into\,v}{\sigma''}$ and
    $\into\, v\in \vinterp[\Delta]{\Fix\,\alpha.A}{n,\sigma''}$.
    \end{itemize}
\end{proof}
